\documentclass[letterpaper,11pt]{article}
\usepackage{fullpage}

\newcommand{\Hbox}[1]{\parbox[c][#1][s]{0pt}{}}

\usepackage{enumitem}
\usepackage{graphicx}
\usepackage{color}
\usepackage{color}
\usepackage{amsmath}
\usepackage{amsfonts}
\usepackage{amssymb} 
\usepackage{amsthm}
\usepackage{multirow}
\usepackage{graphicx}
\usepackage{rotating}
\usepackage{multicol}

\usepackage{algorithm}
\usepackage[noend]{algorithmic}
\newcommand{\keyword}[1]{\textrm{\textbf{#1}}}

\newtheorem{lemma}{Lemma}
\newtheorem{claim}{Claim}
\newtheorem{theorem}{Theorem}

\newcommand{\nit}{\texttt{SV}}
\newcommand{\afu}[1]{\texttt{AFU$_#1$}}
\newcommand{\afuIII}{\texttt{AFU$_3$}}
\newcommand{\afuIV}{\texttt{AFU$_4$}}

\newcommand{\afuk}{\texttt{AFU$_k$}}
\newcommand{\siIII}{\texttt{FFF$_3$}}
\newcommand{\siIV}{\texttt{FFF$_4$}}
\newcommand{\siV}{\texttt{FFF$_5$}}
\newcommand{\siVI}{\texttt{FFF$_6$}}
\newcommand{\siVII}{\texttt{FFF$_7$}}
\newcommand{\sik}{\texttt{FFF$_k$}}
\newcommand{\sic}[1]{\texttt{cFFF$_#1$}}

\newcommand{\sick}{\texttt{cFFF$_k$}}

\def\full#1{{}}

\begin{document}
\title{Clique counting in MapReduce:\\theory and experiments}

\author{
  Irene Finocchi, Marco Finocchi, Emanuele G. Fusco\\ Computer Science Department\\{\it Sapienza} University of Rome\\{\tt\{finocchi, fusco\}@di.uniroma1.it}~~~~~~~{\tt mrcfinocchi@gmail.com}
}

\date{}

\maketitle
\thispagestyle{empty}

\begin{abstract}
We tackle the problem of counting the number of $k$-cliques in large-scale graphs, for any constant $k \ge 3$. Clique counting is essential in a variety of applications, among which social network analysis. Due to its computationally intensive nature, we settle for parallel solutions in the MapReduce framework, which has become in the last few years a {\em de facto} standard for batch processing of massive data sets. We give both theoretical and experimental contributions.

On the theory side, we design the first exact scalable algorithm for counting (and listing) $k$-cliques. Our algorithm uses $O(m^{3/2})$ total space and $O(m^{k/2})$ work, where $m$ is the number of graph edges. This matches the best-known bounds for triangle listing when $k=3$ and is work-optimal in the worst case for any $k$, while keeping the communication cost independent of $k$. We also design a sampling-based estimator that can dramatically reduce the running time and space requirements of the exact approach, while providing very accurate solutions with high probability.

We then assess the effectiveness of different clique counting approaches through an extensive experimental analysis over the Amazon EC2 platform, considering both our algorithms and their state-of-the-art competitors. The experimental results clearly highlight the algorithm of choice in different scenarios and prove our exact approach to be the most effective when the number of $k$-cliques is large, gracefully scaling to non-trivial values of $k$ even on clusters of small/medium size. Our approximation algorithm achieves extremely accurate estimates and large speedups, especially on the toughest instances for the exact algorithms. As a side effect, our study also sheds light on the number of $k$-cliques of several real-world graphs, mainly social networks, and on its growth rate as a function of $k$.

\vspace{8mm}
\noindent {\bf Keywords: }{Clique listing, graph algorithms, MapReduce, parallel algorithms, experimental algorithmics}
\end{abstract}

\newpage

\setcounter{page}{1}

%-------------------------------------------------------------------------------------------
%-------------------------------------------------------------------------------------------
\section{Introduction}

The problem of  counting small subgraphs with specific structural properties in large-scale networks has gathered a lot of interest from the research community during the last few years. Counting -- and possibly listing -- all instances of triangles, cycles, cliques, and other different structures is indeed a fundamental tool for uncovering the properties of complex networks~\cite{milo2002}, with wide-ranging applications that include spam and anomaly detection~\cite{BCDLB06,GKT05}, social network analysis~\cite{HR05,RU12}, and the discovery of patterns in biological networks~\cite{saha2010}. 

Since many interesting graphs have by themselves really large sizes, developing analysis algorithms that scale gracefully on such instances is rather challenging. Designing efficient sequential algorithms is often not enough: even assuming that the input graphs could fit into the memory of commodity hardware,  subgraph counting is a computationally intensive problem and the running times of sequential algorithms can easily become unacceptable in practice. 
To overcome these issues, many recent works have focused on speeding up the computation by exploiting parallelism (e.g., using MapReduce~\cite{DG04}), by working in external memory models~\cite{Vitter08}, or by settling for approximate -- instead of exact -- answers (as done, e.g., in data streaming~\cite{muthu05}). 

In this paper we tackle the problem of counting the number of $k$-cliques in large-scale graphs, for any constant $k\ge 3$. This is a fundamental problem in social network analysis (see, e.g.,~\cite{HR05} -- Chapter 11) and algorithms that produce a census of all cliques are also included in widely-used software packages, such as UCINET~\cite{BEF02}. We present simple and scalable algorithms suitable to be implemented in the MapReduce framework~\cite{DG04} that, together with its open source implementation Hadoop~\cite{hadoop}, has become a {\em de facto} standard for programming massively distributed systems  both in industry and academia. MapReduce indeed offers programmers the possibility to easily run their code on large clusters while neglecting any issues related to scheduling, synchronization, communication, and error detection (that are automatically handled by the system). Computational models for analyzing MapReduce algorithms are described in~\cite{KSV10,GSZ11,PPRSU12}.

$k$-clique counting is a natural generalization of triangle counting, where $k=3$: this is the simplest, non-trivial version of the problem and has been widely studied in the literature. Appendix A describes related works in a variety of models of computation. Focusing on MapReduce, two different exact algorithms for listing triangles have been proposed by Suri and Vassilvitskii~\cite{SV11} and validated on real world datasets. One of these algorithms has been recently extended to count arbitrary subgraphs by Afrati {\em et al.}~\cite{AFU13}, casting it into a general framework based on the computation of multiway joins. A sampling-based randomized approach whose output estimate is strongly concentrated around the true number of triangles (under mild conditions) has been finally described by Pagh and Tsourakakis in~\cite{PT12}. We remark that subgraph counting becomes more and more computationally demanding as the number $k$ of nodes in the counted subgraph gets larger, due to the combinatorial explosion of the number of candidates. This is especially true for $k$-cliques, as we will also show in this paper, since certain real world graphs (such as social networks) are characterized by high clustering coefficients and very large numbers of small cliques.

\vspace{-2mm}
\paragraph{Our results.} The contribution in this paper is two-fold and includes both theoretical and experimental results. On the theory side, we design and analyze the first scalable exact algorithm for counting (and listing) $k$-cliques as well as sampling-based approximate solutions. In more details:
\begin{itemize}[leftmargin=0.4cm,rightmargin=0cm]
\itemsep0em
\vspace{-2mm}

\item Our exact $k$-clique counting algorithm uses $O(m^{3/2})$ total space and $O(m^{k/2})$   work, where $m$ is the number of graph edges. The local space and the local running time of mappers and reducers are $O(m)$ and $O(m^{(k-1)/2})$, respectively. For $k=3$, total space and work match the bounds for triangle counting achieved in~\cite{SV11}. Similarly to the multiway join algorithm from~\cite{AFU13}, our algorithm is work-optimal in the worst case for any $k$. However, our total space is proportional to $m^{3/2}$ regardless of $k$: this means that, differently from~\cite{AFU13}, the communication cost does not grow when $k$ gets larger.

\item Our sampling-based estimators reduce dramatically running time and local space requirements of the exact approach, while providing very accurate solutions \full{with high probability} w.h.p. These algorithms can be proved to belong to the class $\mathcal{MRC}$~\cite{KSV10} for a suitable choice of the probability parameters. For $k=3$, our concentration results require weaker conditions than the triangle counting algorithm from~\cite{PT12}.

\end{itemize}
\vspace{-1mm}

\noindent To assess the effectiveness of different clique counting approaches, we conducted a thorough experimental analysis over the Amazon EC2 platform, considering both our algorithms and those described in~\cite{AFU13} and~\cite{SV11}. We used  publicly available real data sets taken from the SNAP graph library~\cite{SNAP} and synthetic graphs generated according to the preferential attachment model~\cite{BA99}. As a side effect, our experimental study also sheds light on the number of $k$-cliques of several SNAP datasets and on its growth rate as a function of $k$. The outcome of our experiments can be summarized as follows:

\begin{itemize}[leftmargin=0.4cm,rightmargin=0cm]
\itemsep0em
\vspace{-2mm}

\item Even for small values of $k$, some of the input graphs contain a number $q_k$ of $k$-cliques that can  be in the order of tens or hundreds of trillions (recall that a trillion is $10^{12}$): in these cases, the output of the $k$-clique listing problem could easily require Terabytes or even Petabytes of storage (assuming no compression). As $k$ increases, we witnessed different growth rates of $q_k$ for different graph instances: while on some graphs $q_{k+1}<q_k$ (even for small values of $k$), on other instances $q_{k+1}\gg q_k$, up to two orders of magnitude in our observations. Graphs with a quick growth of the number of $k$-cliques represent particularly tough instances in practice.

\item Among the  exact algorithms considered in our analysis, our approach proves to be the most effective when the number of $k$-cliques is large. There are cases where it is outperformed by the triangle counting algorithm of~\cite{SV11} and by the multiway join algorithm of~\cite{AFU13}: this happens either for $k=3,4$ or on ``easy'' instances where parallelization does not pay off (we observed that in these cases a simple sequential algorithm would be even faster). On the tough instances, however, our algorithm can gracefully scale as $k$ gets larger, differently from the multiway join approach~\cite{AFU13}. We provide a theoretical justification of these experimental findings.

\item Our approximate algorithms exhibit rather stable running times on all graphs for the considered values of $k$, and make it possible to solve in a few minutes instances that were impossible to be solved exactly. The quality of the approximation is extremely good: the error is around $0.08\%$ on average, with more accurate estimates on datasets that are more challenging for the exact algorithms. The variance across different executions, even on different clusters, also appears to be negligible.

\end{itemize}
\vspace{-1mm}

\noindent Overall, the experiments show the practical effectiveness of our algorithms even on clusters of small/ medium size, and suggest their scalability to larger clusters. The full experimental package is available at {\tt https://github.com/CliqueCounter/QkCount/} for the purpose of repeatability. 
 
\full{
\paragraph{Organization of the paper.} Section~\ref{se:preliminaries} gives some theoretical preliminaries and briefly reviews MapReduce and the computational model proposed in~\cite{KSV10}. Section~\ref{se:exact} presents our exact algorithm, analyzing its communication and computation costs. Section~\ref{se:setup} describes our experimental framework, including the evaluated algorithms and the Amazon cluster setup. Section~\ref{se:experiments} summarizes our main experimental findings on exact algorithms. Our sampling-based estimator is described and analyzed, both theoretical and experimentally, in Section~\ref{se:approximate}.
}

\vspace{-2mm}
%-------------------------------------------------------------------------------------------
%-------------------------------------------------------------------------------------------
\section{Preliminaries}
\label{se:preliminaries}

\full{We will make use of the following notation. For a set $V$, $V^i$ denotes the $i$-th Cartesian product, i.e., $V^i=V\times V\times \ldots \times V$, $i$ times.}
Throughout this paper we denote by $q_k$ the number of cliques on $k$ nodes.
For a given graph $G$ and any node $u$ in $G$, $\Gamma(u)$ is the set of neighbors of node $u$ ($u$ is not included); moreover $d(u)=|\Gamma(u)|$.
We define a total order $\prec$ over the nodes of $G$ as follows: $\forall x,y \in V(G)$, $x \prec y$ iff $d(x)<d(y)$ or $d(x)=d(y)$ and  $x<y$ (we assume nodes to have comparable labels).
Denote by $\Gamma^+(u) \subseteq \Gamma(u)$ the {\em high-neighborhood} of node $u$, i.e., the set of neighbors $x$ of $u$ such that $u \prec x$; symmetrically, $\Gamma^-(u)=\Gamma(u)\setminus\Gamma^+(u)$. \full{is the set of neighbors $x$ of node $u$ such that $x \prec u$. }
Given two graphs $G(V,E)$ and $G_1(V_1,E_1)$, $G_1$ is
a subgraph of $G$ if $V_1\subseteq V$ and $E_1\subseteq E$. $G_1$ is an {\em induced subgraph} of $G$ if, in addition to the above conditions, for each $u,v\in V_1$ it also holds: $(u, v) \in E_1$ if and only if $(u, v)\in E$. We denote the subgraph induced by the high-neighborhood $\Gamma^+(u)$ of a node $u$ as $G^+(u)$.
Our algorithms do not require graphs to be connected. However, since their input is given as a set of edges and isolated nodes are irrelevant for clique counting, we can assume \full{the number $n$ of nodes to be at most twice the number $m$ of edges, i.e.,} $n\le 2m$.
Moreover, we assume the endpoints of each edge to be labeled with their degree (the degree of each node can be precomputed in MapReduce very efficiently~\cite{KSV10}). Additional preliminary results used in our proofs are given in Appendix B.

\vspace{2mm}
\noindent{\bf MapReduce.} 
A MapReduce program is composed of (usually a small number of) rounds.
Each round is conceptually divided into three consecutive phases: {\em map}, {\em shuffle}, and {\em reduce}.
Input/output values of a round, as well as intermediate data exchanged between mappers and reducers, are stored as $\langle key; value \rangle$ pairs.
In the map phase pairs are arbitrarily distributed among mappers and a programmer-defined map function is applied to each pair. Mappers are stateless and process each input pair independently from the others. The shuffle phase is transparent to the programmer: during this phase, the intermediate output pairs emitted by the mappers are grouped by key. All pairs with the same key are then sent to the same reducer, that will process each of them by executing a programmer-defined reduce function.
Parallelism comes from concurrent execution of mappers as well as reducers.
The authors of~\cite{KSV10} made an effort to pinpoint the critical aspects of efficient MapReduce algorithms. In particular: (1) the memory used by a single mapper/reducer should be sublinear with respect to the total input size (this allows to exclude trivial algorithms that simply map the whole input to a single reducer, which then solves the problem via a sequential algorithm); (2) the total number of machines available should be sublinear in the data size; (3) both the map and the reduce functions should run in polynomial time with respect to the original input length.
\full{
The authors of~\cite{KSV10} made an effort to pinpoint the critical aspects of efficient MapReduce algorithms.
In particular, they considered the following aspects:
\begin{description}
\item[Memory:] the memory used by a single mapper or reducer should be sublinear with respect to the total input size. This allows to exclude trivial algorithms that simply map the whole input to a single reducer, which then solves the problem via a sequential algorithm.
\item[Machines:] the total number of machines available should be sublinear in the data size.
\item[Time:]  both the map and the reduce functions should run in polynomial time with respect to the original input length.
\end{description}
}
The model also requires programs to be composed of a polylogarithmic number of rounds, since shuffling is a time consuming operation, and to have a total memory usage (which coincides with the communication cost from~\cite{AFU13} for the purpose of this paper) that grows substantially less than quadratically with respect to the input size. Algorithms respecting these conditions are said to belong to the class $\mathcal{MRC}$.

\full{Besides the model proposed in~\cite{KSV10}, which we will use throughout this paper, other efforts to formalize the analysis of MapReduce algorithms have been proposed in \cite{GSZ11} and~\cite{PPRSU12}. Other efforts to formalize the analysis of MapReduce algorithms have been proposed in \cite{GSZ11} and~\cite{PPRSU12}.}

\vspace{-3mm}
%==============================================================
\section{Exact counting}
\label{se:exact}

Our algorithms use the total order $\prec$ to decide which node of a given clique $Q$ is responsible for counting $Q$. In particular, $Q$ is counted by its \emph{smallest} node, i.e., the node $u\in Q$ such that $u\prec x$, for all $x\in Q\setminus\{u\}$.
At a high level, the strategy is to split the whole graph in many subgraphs, namely the subgraphs $G^+(u)$ induced by $\Gamma^+(u)$, for each $u\in V(G)$, and count the cliques in each subgraph independently (both nodes and edges of $G$ can appear in more than one subgraph).
Our counting algorithm, called \sik,  works in three rounds (see Algorithm~\ref{fi:qk-count} for the pseudocode):

%-----------------------------------------------------------------------------------
\begin{algorithm}[t]
\begin{multicols}{2}
\begin{small}
\smallskip
\keyword{Map 1:} input $\langle(u,v); \emptyset  \rangle$

\begin{algorithmic}[0]
\STATE {\bf{if}}~{$u\prec v$} {\bf{then}} emit $\langle u;v\rangle$
\end{algorithmic}

\smallskip
\keyword{Reduce 1:} input $\langle u; \Gamma^+(u) \rangle$

\begin{algorithmic}[0]
\STATE {\bf{if}}~{$|\Gamma^+(u)| \ge k-1$} {\bf{then}} emit $\langle u; \Gamma^+(u)\rangle$
\end{algorithmic}

\smallskip
\keyword{Map 2:} input $\langle u; \Gamma^+(u) \rangle$ or $\langle(u,v); \emptyset  \rangle$

\begin{algorithmic}[0]
\IF{input of type $\langle(u,v); \emptyset  \rangle$ and $u\prec v$}
\STATE emit $\langle (u,v); \$) \rangle$
\ENDIF
\IF{input of type $\langle u; \Gamma^+(u) \rangle$}
\FOR{\keyword{each} $x_i,\,x_j \in \Gamma^+(u)$ s.t. $x_i \prec x_j$}
%\FOR{\keyword{each} $x_i \in \Gamma^+(u)$ }
%	\FOR{\keyword{each} $x_j \in \Gamma^+(u)$ s.t. $x_i \prec x_j$}
%		\STATE emit $\langle (x_i,x_j); u \rangle$
%	\ENDFOR
		\STATE emit $\langle (x_i,x_j); u \rangle$
\ENDFOR
\ENDIF
\end{algorithmic}
\newpage

\smallskip
\keyword{Reduce 2:} input $\langle (x_i, x_j);  \{ u_1, \ldots, u_k \} \cup \mbox{\$}\rangle$ 

\begin{algorithmic}[0]
\IF{input contains \$}
	\STATE emit $\langle (x_i, x_j);   \{ u_1, \ldots, u_k \}  \rangle$
\ENDIF
\end{algorithmic}

\smallskip
\keyword{Map 3:} input $\langle (x_i,x_j);  \{ u_1, \ldots, u_k \} \rangle$

\begin{algorithmic}[0]
\FOR{$h \in [1,k]$}
	\STATE emit $\langle u_h; (x_i, x_j) \rangle$
\ENDFOR
\end{algorithmic}

\smallskip
\keyword{Reduce 3:} input $\langle u;  G^+(u) \rangle$ 
\begin{algorithmic}[0]
\STATE let $q_{u,k-1}=$ number of $(k-1)$-cliques in $G^+(u)$
\STATE emit $\langle u; q_{u,k-1} \rangle$
\end{algorithmic}
\end{small}
\end{multicols}
\vspace{-3mm}
\caption{: \sik}\label{fi:qk-count}
\end{algorithm}

\begin{description}[leftmargin=0.4cm,rightmargin=0cm]
\itemsep0em
\vspace{-2mm}

\item[Round 1:] \emph{high-neighborhood computation.} The computation of $\Gamma^+(u)$, for all nodes $u$ in $G$, exploits the degree information attached to the edges; mappers emit the pair $\langle x; y\rangle$ for each edge $(x,y)$ such that $x\prec y$, thus allowing the reduce instance with key $u$ to aggregate all nodes $x\in \Gamma^+(u)$.

\item[Round 2:] \emph{small-neighborhoods intersection.} 
The aim of the round is to associate each edge $(x,y)$ with $\Gamma^-(x)\cap\Gamma^-(y)$, i.e., with the set of nodes $u$ such that $G^+(u)$ contains $(x,y)$. 
This is done as follows. The map instance with input  $\langle u;\Gamma^+(u)\rangle$  emits a pair $\langle (x,y); u \rangle$ for each pair $(x,y)\in \Gamma^+(u)\times\Gamma^+(u)$ such that $x\prec y$. Besides the output of round 1, similarly to~\cite{SV11} mappers are fed with the original set of edges and emit a pair $\langle (x,y); \$\rangle$ for each edge $(x,y)$ with $x\prec y$.
This allows the reduce instance with key $(x,y)$ to check whether $(x,y)$ is an edge by looking for symbol $\$$ among its input values. At the same time this instance would receive the set $\Gamma^-(x)\cap\Gamma^-(y)$, which is  exactly the set of nodes $u$ needing edge $(x,y)$ to construct $G^+(u)$.

\item[Round 3:] ($k-1$)-\emph{clique counting in high-neighborhoods.} For each node $u$,  count the number of $k$-cliques for which $u$ is responsible. \full{This is done as follows.} Map instances correspond to graph edges. The map instance with key  $(x,y)$ emits a pair $\langle u; (x,y)\rangle$ for each node $u\in \Gamma^-(x)\cap\Gamma^-(y)$. 
After shuffling, the reduce instance with key $u$ receives as input the whole list of edges between nodes in $\Gamma^+(u)$. Hence, it can reconstruct the subgraph $G^+(u)$ induced by the high-neighbors of $u$ and, by counting the $(k-1)$-cliques in this graph, can compute locally the number of $k$-cliques for which $u$ is responsible. 
\end{description}
\vspace{-2mm}

\noindent Notice that the round 3 reducers could be easily modified to output, for each node $v$, the number of cliques in which $v$ is contained, so that the overall number of cliques containing $v$ could be obtained by summing up the contributions from each subgraph $G^+(u)$. The following theorem, proved in Appendix B, analyzes work and space usage of \sik:

\begin{theorem}
\label{th:qk-count}
Let $G$ be a graph and let $m$ be the number of its edges.
Algorithm \sik\ counts the number of $k$-cliques of $G$ using $O(m^{3/2})$ total space and $O(m^{k/2})$  work. The local space and the local running time of mappers and reducers are $O(m)$ and $O(m^{(k-1)/2})$, respectively.
\end{theorem}

\full{\vspace{-4mm}
\paragraph{Discussion.}}
With respect to the requirements defined in~\cite{KSV10}, algorithm \sik\ does not fit in the class $\mathcal{MRC}$ due to the local space requirements of reduce 2, map 3, and reduce 3 instances. These are linear in $n$ or $m$ whereas the $\mathcal{MRC}$  class requires them to be in $O(m^{1-\varepsilon})$, for some small constant $\varepsilon>0$. However, as we will see from the experimental evaluation we performed, local memory did not show any criticality in practice, whereas local complexity (which is almost completely neglected in the class $\mathcal{MRC}$) and global work proved to be the real challenge, since they can grow significantly as $k$ gets larger.
Our approximate algorithms presented in Section~\ref{se:approximate} overcome these issues.

Similarly to the triangle counting algorithms from~\cite{SV11} and to the multiway join subgraph counting algorithm from~\cite{AFU13}, cast to the specific case of cliques (in short, \afuk), the work of \sik\ is optimal with respect to the clique listing problem. Moreover, the total space of  \sik\ is $\Theta(m^{3/2})$ regardless of $k$, matching the triangle counting bound of the {\tt Node\,Iterator\,++} algorithm from~\cite{SV11} when $k=3$. 
Conversely, both \afuk\ and the straightforward generalization to $k$-cliques of the {\tt Partition} algorithm from~\cite{SV11}  have a communication cost that depends both on a number $b$ of {\em buckets}, chosen as a parameter, and on the number $k$ of clique nodes (see \cite{AFU13} and \cite{SV11} for details). In \afuk, reducers are identified by $k$-tuples of buckets $\langle b_i\le b_2\le \ldots \le b_k\rangle$. Each edge corresponds to a pair $\langle i,j\rangle$ of buckets (those to which its endpoints are hashed) and is sent to all the reducers whose $k$-tuple contains both $i$ and $j$. Each edge is thus replicated $\Omega(b^{k-2})$ times, for an overall communication cost $\Theta(m \cdot b^{k-2})$. Similar arguments apply to the generalization of {\tt Partition}. We will see the implications in Section~\ref{se:experiments}. We also notice that, differently from \afuk\ and {\tt Partition}, algorithm \sik\ needs no parameter tuning.

\vspace{-2mm}

%-------------------------------------------------------------------------------------------
%-------------------------------------------------------------------------------------------
\section{Experimental setup}
\label{se:setup}

\vspace{-2mm}
\noindent{\bf Algorithms and implementation details.} Besides algorithm \sik, we included in our test suite the {\tt Node\,Iterator\,++} triangle counting algorithm from~\cite{SV11} (called \nit) and the one-round subgraph counting algorithm \afuk\ based on multiway joins~\cite{AFU13}, cast to $k$-cliques. We did not consider the {\tt Partition} algorithm from~\cite{SV11} because \afuIII\ is its optimized version.
Both the reduce 3 instances of \sik\ and the reducers of \afuk\ use as a subroutine an efficient clique counting algorithm based on neighborhood intersection, inspired by the fastest (optimal) triangle listing algorithm described in~\cite{OB14}.  In the case of \afuk, we took care of exploiting bucket orderings to discard as soon as possible $k$-cliques that do not fit in the $k$-tuple of a given reducer. According to our tests, this results in slightly faster running times w.r.t. to the plain version.
All the implementations have been realized in Java using Hadoop 2.2.0. The code is instrumented so as to collect detailed statistics of map/reduce instances at each round, including sizes of the subgraphs involved in a computation (e.g., $|\Gamma^+(u)|$, $|\Gamma^-(x)\cap\Gamma^-(y)|$, and $|G^+(u)|$) and detailed running times. \full{Due to the extremely large size of the log files and in order to prevent Heisenberg effects, all the results reported in this paper were obtained with instrumentation disabled, unless otherwise noticed.} The algorithms have been tested in a variety of settings, using different parameter choices, instance families, and cluster configurations, as described below.  \full{The full package is available at: {\tt https://github.com/CliqueCounter/QkCount/} for the sake of repeatability.}

\begin{table}[t]
\begin{center}
\begin{footnotesize}
\centering\begin{tabular}{|clllllll|}\hline
&  \multicolumn{1}{|c}{$n$ ($=q_1$)} & \multicolumn{1}{c}{$m$ ($=q_2$)} & \multicolumn{1}{c}{$q_3$} & \multicolumn{1}{c}{$q_4$} & \multicolumn{1}{c}{$q_5$} & \multicolumn{1}{c}{$q_6$} & \multicolumn{1}{c|}{$q_7$} \\\hline\hline 
%\multirow{1}{2cm}{
%&  \multicolumn{1}{c}{}  &  &  &  &  &  &  \\[-2.5ex]
\multicolumn{1}{|c|}{\texttt{citPat} }& \Hbox{13pt}$3.8\times10^6$ & $1.6\times10^7$ & $7.5\times10^6$ & $3.5\times10^6$ & $3.0\times10^6$ & $3.1\times10^6$  & $1.9\times10^6$\\
\multicolumn{1}{|c|}{\texttt{youTube}} &  $1.1\times 10^5$ & $3.0\times 10^6$  &$3.0\times10^6$ & $5.0 \times 10^6$  & $7.2\times10^6$ &$8.4\times 10^6$  & $8.0\times 10^6$ \\
\multicolumn{1}{|c|}{\texttt{locGowalla}} &  $2.0\times10^5$ & $9.5\times10^5$ & $2.7\times10^6$  & $6.1\times10^6$  & $1.5\times10^7$ & $2.9\times10^7$ & $4.8\times10^7$\\
\multicolumn{1}{|c|}{\texttt{socPokec} }&  $1.6\times10^6$  & $2.2\times10^{7}$   & $3.3\times10^7$ & $4.3\times10^7$ & $5.3\times10^7$ & $6.5\times10^7$ & $8.4\times10^7$ \\
\multicolumn{1}{|c|}{\texttt{webGoogle}} &  $8.7\times10^5$ & $4.3\times10^6$ & $1.3\times10^7 $ & $3.4\times10^7$ & $1.0\times10^8$  & $2.5\times10^8$ & $6.0\times10^8$ \\
\multicolumn{1}{|c|}{\texttt{webStan}} &  $2.8\times10^5$ & $2.0\times10^6$ & $1.1\times10^7$ & $7.9\times10^7$ & $6.2\times10^8$ & $4.9\times10^9$ & $3.5\times10^{10}$ \\
\multicolumn{1}{|c|}{\texttt{asSkit} }&  $1.7\times10^6$ & $1.1\times10^7$ & $2.9\times10^7$ & $1.5\times10^8$ & $1.2\times10^9$ & $9.8\times10^9$  & $7.3\times10^{10}$ \\
\multicolumn{1}{|c|}{\texttt{orkut} }&  $3.1\times10^6$ & $1.2\times10^8$ & $6.3\times10^8$ & $3.2\times10^9$ & $1.6\times10^{10}$& $7.5\times10^{10}$ & $3.5\times10^{11}$ \\
\multicolumn{1}{|c|}{\texttt{webBerkStan}} &  $6.8\times10^5$ & $6.6\times10^6$ & $6.5\times10^7$ & $1.1\times10^9$ & $2.2\times10^{10}$ & $4.6\times10^{11}$& $9.4\times10^{12}$ \\
\multicolumn{1}{|c|}{\texttt{comLiveJ}} &  $4.0\times10^6$ & $3.5\times10^7$ & $1.8\times10^8$ & $5.2\times10^9$ & $2.5\times10^{11}$ & $1.1\times10^{13}$ & {\it 4.4\,$\times$\,10$^{14}$} \\
\multicolumn{1}{|c|}{\texttt{socLiveJ1}} &  $4.8\times10^6$ & $4.3\times10^7$ & $2.9\times10^8$ & $9.9\times10^9$ & $4.7\times10^{11}$ & $2.1\times10^{13}$ & {\it 8.6\,$\times$\,10$^{14}$}\\
\multicolumn{1}{|c|}{\texttt{egoGplus}} &  $1.1\times 10^5$ & $1.2\times10^7$  & $1.1\times10^9$  & $7.8\times10^{10}$ & $4.7\times10^{12}$  & {\it 2.4\,$\times$\,10$^{14}$}& {\it 1.1\,$\times$\,10$^{16}$}
\\\hline
\end{tabular}
\end{footnotesize}
\vspace{-1mm}
\caption{Benchmark statistics: order of magnitude of $n$, $m$, and $q_k$ (numbers of nodes, edges, and $k$-cliques, respectively) for $k\in[3,7]$. Clique numbers in {\it italic} are estimates obtained by the algorithm described in Section~\ref{se:approximate}. \full{The table in Appendix C shows the exact values of $n$, $m$, and $q_k$, the storage in MB, and the ratio $q_{k+1}/q_k$ that gives the clique growth rate as a function of $k$.} The table in Appendix C shows the exact values and the clique growth rates.}
\vspace{-7mm}
\label{fig:data-sets}
\end{center}
\end{table}

\vspace{2mm}
\noindent{\bf Data sets.} We used several real-world graphs from the SNAP graph library~\cite{SNAP} as well as synthetic graphs generated according to the preferential attachment model~\cite{BA99}. We preprocessed all graphs so that they are undirected and each edge endpoint is associated with its degree. Degree computation is a common step to both \nit\ and \sik, and can be done very easily and quickly in MapReduce~\cite{KSV10}. Throughout the paper we report on the results obtained for a variety of online social networks (called \texttt{orkut}, \texttt{socPokec}, \texttt{youTube}, \texttt{locGowalla}, \texttt{socLiveJ1}, \texttt{comLiveJ}, \texttt{egoGplus}), Web graphs (\texttt{webBerkStan}, \texttt{webGoogle}, \texttt{webStan}), an Internet topology graph (\texttt{asSkit}), and a citation network among US Patents (\texttt{citPat}).
The main characteristics of these datasets are summarized in Table~\ref{fig:data-sets}. Notice that only the number  $q_3$ of triangles was available from~\cite{SNAP} before our study. With respect to $q_3$, the number of $k$-cliques for larger values of $k$ can grow considerably, up to the order of tens or even hundreds of trillions. 

\vspace{2mm}
\noindent{\bf Platform.} The experiments have been carried out on three different Amazon EC2 clusters, running Hadoop 2.2.0. Besides the master node, the three clusters included 4, 8, and 16 worker nodes, respectively, devoted to both Hadoop tasks and the HDFS. We used Amazon EC2 {\tt m3.xlarge} instances, each providing 4 virtual cores, \full{(based on Intel Xeon E5-2670 v2 Sandy Bridge processors)} 7.5 GiB of main memory, and a 32 GB solid state disk. We set the number of reduce tasks to match the number of virtual cores in each cluster and disabled speculative execution. We also modified the memory requirements of the containers in order to improve load balancing on the cluster. The precise Hadoop configuration used on the Amazon clusters is provided with our experimental package available on \texttt{github}.

%-------------------------------------------------------------------------------------------
%-------------------------------------------------------------------------------------------
\begin{table}[t]
\centering
\begin{footnotesize}
\begin{tabular}{|c|rrr|rr|rr|rr|rr|}\hline
%\multicolumn{1}{|c}{} &
%\multicolumn{11}{|c|}{Running time (minutes\,:\,seconds)} \\\hline\hline
& \nit  &\siIII & \afu{3}  & \siIV & \afu{4} & \siV & \afu{5}  & \siVI & \afu{6} & \siVII &\afu{7} \\\hline\hline
\texttt{citPat} & 2:44 & 3:22 & 2:23 & 3:11 & 3:11 & 3:13 & 2:18 & 3:13 & 2:19 & 3:09 & 2:24 \\
\texttt{youTube} & 2:06 & 2:04 & 1:25 & 2:39 & 1:41 & 2:34 & 1:33 & 2:36 & 1:39 & 2:38 & 1:49 \\
\texttt{locGowalla} & 2:36 & 3:08 & 1:18 & 3:04 & 1:21 & 3:02 &1:30 &3:04 &1:24 &3:03 & 1:30 \\
\texttt{socPokec} & 4:03 & 4:15 & 2:18 & 4:02 & 2:29 & 4:13 & 2:39 & 4:15 & 2:51 & 4:09 & 3:02 \\
\texttt{webGoogle} & 2:13 & 2:44 & 1:23 & 2:43 & 1:27 & 2:43 & 1:32 & 2:40 & 1:40 & 2:40 & 1:52 \\
\texttt{webStan} & 2:02 & 2:39 & 1:15 & 2:29 & 1:27 & 2:37 & 2:06 & 2:36 & 4:00 & 2:05 &14:12 \\
\texttt{asSkit} &  2:44 & 3:14 & 1:43 & 3:17 & 2:59 & 3:18 & 5:34 & 3:14 & 25:30 & 4:12 & $>$40 \\
\texttt{orkut} & 30:07 &  24:00 & 8:21 & 23:08 & 20:17 & 23:10 & $>$50 & 23:22 & $>$50 & 28:08 & - \\
\texttt{webBerkStan} & 2:28 & 3:00 & 1:37 & 3:01 & 2:53 & 3:08 & 8:24 & 4:56 & $>$30 & 50:17 & - \\
\texttt{comLiveJ} & 5:31  & 5:31 & 2:53 & 5:24  &  4:06 & 6:13 & 14:02 & 41:22 & $>$170 & - & - \\
\texttt{socLiveJ1} & 6:36 & 6:33 & 3:14 & 6:43 & 5:10 & 7:51 & 23:35 & 86:34&  $>$180 &  - & -\\
\texttt{egoGplus} & 22:54 & 17:19 & 2:06 & 17:54 & 16:55 & 39:01 & $>$90 & - & -  & - & -\\
\hline
\end{tabular}
\end{footnotesize}
\vspace{-1mm}
\caption{Running time  (minutes:seconds) of the algorithms on a 16-node cluster with 64 total cores. To minimize the costs, we killed some executions of \afuk\ that took more than twice the time of \sik. For the sake of comparison, the running times of algorithm \nit\ reported in~\cite{SV11} are 1.90, 1.77, and 5.33  minutes on \texttt{asSkit}, \texttt{webBerkStan}, and \texttt{socLiveJ1}, respectively, on a 1636-node cluster. 
}
%with \sik\ was large enough (in those cases we report the killing times).}
\label{fig:runtime-15nodes}
\end{table} 

\vspace{-3mm}

%-------------------------------------------------------------------------------------------
\section{Computational experiments}
\label{se:experiments}

\vspace{-1mm}

In this section we summarize our main experimental findings. We first present results obtained on a 16-node Amazon cluster, analyzing the effects of $k$ on the performance of the algorithms, and we then address scalability issues on different cluster sizes. Our experiments account for more than 60 hours of computation over the EC2 platform. Table~\ref{fig:runtime-15nodes} is the main outcome of this study, showing the running times of all the evaluated algorithms for $k\le 7$ on the SNAP graphs ordered by increasing $q_7$ (see Appendix C). Results for synthetic instances were consistent with real datasets and are not reported. 

\vspace{2mm}
\noindent{\bf Triangle counting: the costs of rounds.} Since the overhead of setting up a round -- including shuffling -- is non-negligible in MapReduce, the one-round \afuIII\ algorithm is always much faster than \nit\ and \siIII, which respectively require  two and three rounds. \siIII\ is slower than \nit\ on most datasets, but faster on \texttt{orkut} and \texttt{egoGplus}, which have the largest number of triangles (see also Table~\ref{fig:data-sets}). We conjectured that this may be due to the early computation of length-$2$ paths performed in \nit\ by the round 1 reducers (see~\cite{SV11} for details), which uselessly increases the communication cost of round 1. To test our hypothesis, we engineered a variant of \nit\ that delays $2$-path computation to the map phase of round 2. The variant showed largely improved running times, solving \texttt{orkut} and \texttt{egoGplus} in 22 and 12 minutes (instead of 30 and 23, resp.) and being faster than \siIII\ on all benchmarks.

\vspace{2mm}
\noindent{\bf Running time analysis for $k\ge 4$.} Algorithm \siIV\ can compute the number of $4$-cliques within roughly the same time required to count triangles. \afuIV\ remains faster than \siIV, but is always slower than \afuIII: see, in particular, \texttt{orkut} and \texttt{egoGplus}. In general, when $k\ge 5$, \sik\ proves to be more and more effective than \afuk\ and its running times scale gracefully with $k$, especially on the most difficult instances characterized by a steep growth of the number of $k$-cliques (\texttt{asSkit}, the LiveJournal networks, \texttt{orkut},  \texttt{webBerkStan}, and \texttt{egoGplus}).
To explain this behavior, recall that \afuk\ requires to choose the number $b$ of buckets that, together with $k$, determines the number of reducers. The communication cost, as observed in Section~\ref{se:exact}, grows as $\Theta(m \cdot b^{k-2})$, which in practice calls for small values of $b$: if $b$ is too large, even a small input graph could quickly grow to Terabytes of disk usage for moderate values of $k$. On the other hand, since the local running times of reducers are inversely proportional to $b$, if $b$ is too small w.r.t.~$k$ the worst-case reducers incur high running times (if, e.g., $k>b/2$, some $k$-tuple contains more than half of the buckets and the corresponding reducer will receive more than half of the total number of edges). The running times of such reducers  (as well as their memory requirements) remain comparable to that of the sequential algorithm. The practical implication of this communication/runtime tradeoff is that the best performance of \afuk\ have been obtained within our experimental setup using rather small values of $b$. We nevertheless performed experiments with different choices of $b$: as suggested in~\cite{AFU13}, we first set $b$ so that ${b+k-1}\choose k$  is as close as possible to the available reducers, and we then considered a variety of different values. Table~\ref{fig:runtime-15nodes} always reports the fastest running time that we could obtain by separately tuning $b$ for each instance and $k$. 

The $\Theta(m^{3/2})$ communication cost of \sik\  proved to be a bottleneck only in a few datasets, while the actual clique enumeration remains the most time-consuming phase. However, the running times of reducers are much shorter, not only in theory but also in practice, than the application of a sequential algorithm to the whole graph: hence, using more rounds results in poor performance only when the overall task has a short duration, but quickly outperforms both \afuk\ and sequential algorithms on the most demanding graphs and as $k$ increases. Insights on round analysis are given below.

%------------------------------------------
\begin{figure*}[t]
\centering
\begin{tabular}{ccc}
 \hspace{-5mm}
 \includegraphics[scale=0.43]{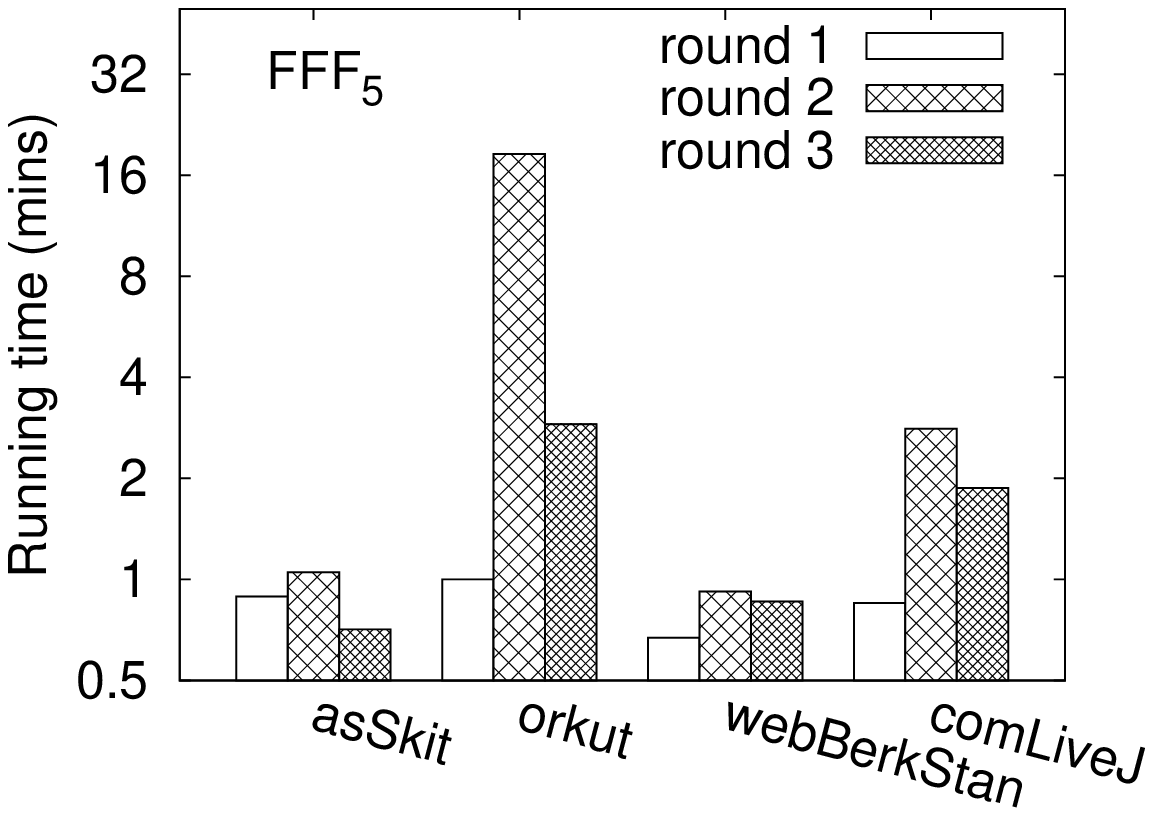} \hspace{-4.5mm} &
 \includegraphics[scale=0.43]{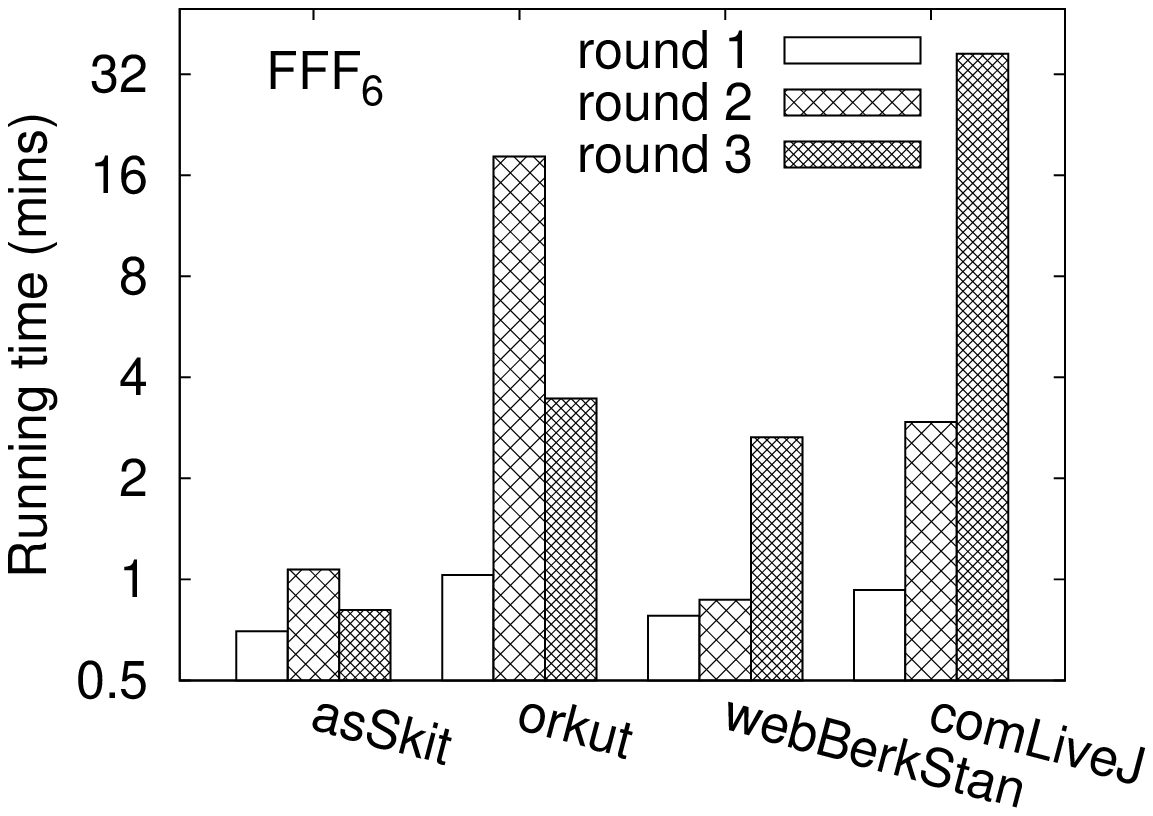} \hspace{-4.5mm} &
 \includegraphics[scale=0.43]{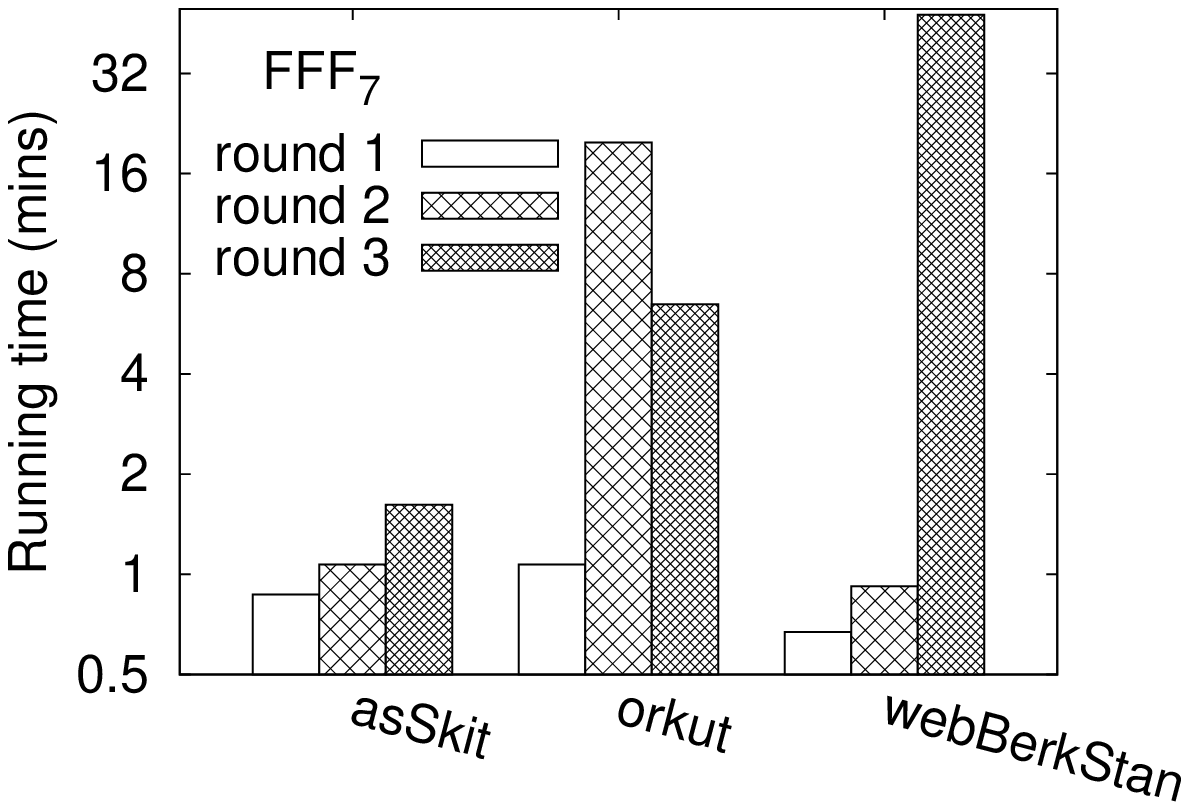} \hspace{-4.5mm}
\end{tabular}
\vspace{-5mm}
\caption{Round-by-round running times of \sik\ on four representative datasets for $k\in[5,7]$.}
\vspace{-3mm}
\label{fig:roundAnalysis}
\end{figure*}

\vspace{2mm}
\noindent{\bf Round-by-round analysis of \sik.} In Figure~\ref{fig:roundAnalysis} we compare the running times of each round of \sik\ on a selection of benchmarks. Round 1 is typically negligible, regardless of the benchmark and of the value of $k$. Round 2, which computes $2$-paths and small-neighborhoods intersections, is the most expensive step for $k\le 5$, but its running time can only decrease when $k$ grows (due to the test $|\Gamma^+(u)|\ge k-1$ performed by reduce 1 instances, see Algorithm~\ref{fi:qk-count}). 
Round 3 becomes more and more expensive as $k$ gets larger, and dominates the running time on {\tt webBerkStan} and {\tt comLiveJ} already for $k=6$. This confirms the intuition supported by our theoretical analysis: computing $(k-1)$-cliques on the subgraphs $G^+(u)$ induced by high-neighborhoods can be rather time-consuming and becomes the dominant operation as $k$ gets larger. Figure~\ref{fig:round5}a shows the cumulative distribution of $|G^+(u)|$, focusing on reduce 3 instances that required more than $100$ ms: notice that a constant fraction of nodes has rather large induced subgraphs (e.g., in {\tt egoGplus} about 5000 nodes have high neighborhoods with a number of edges in-between $2^{16}$ and $2^{18}$, which is the largest $|G^+(u)|$).

%------------------------------------------
\begin{figure*}[t]
\centering
\begin{tabular}{cccc}
 \hspace{-5mm}
  \raisebox{20mm}{(a)} & \hspace{-3mm} \includegraphics[scale=0.50]{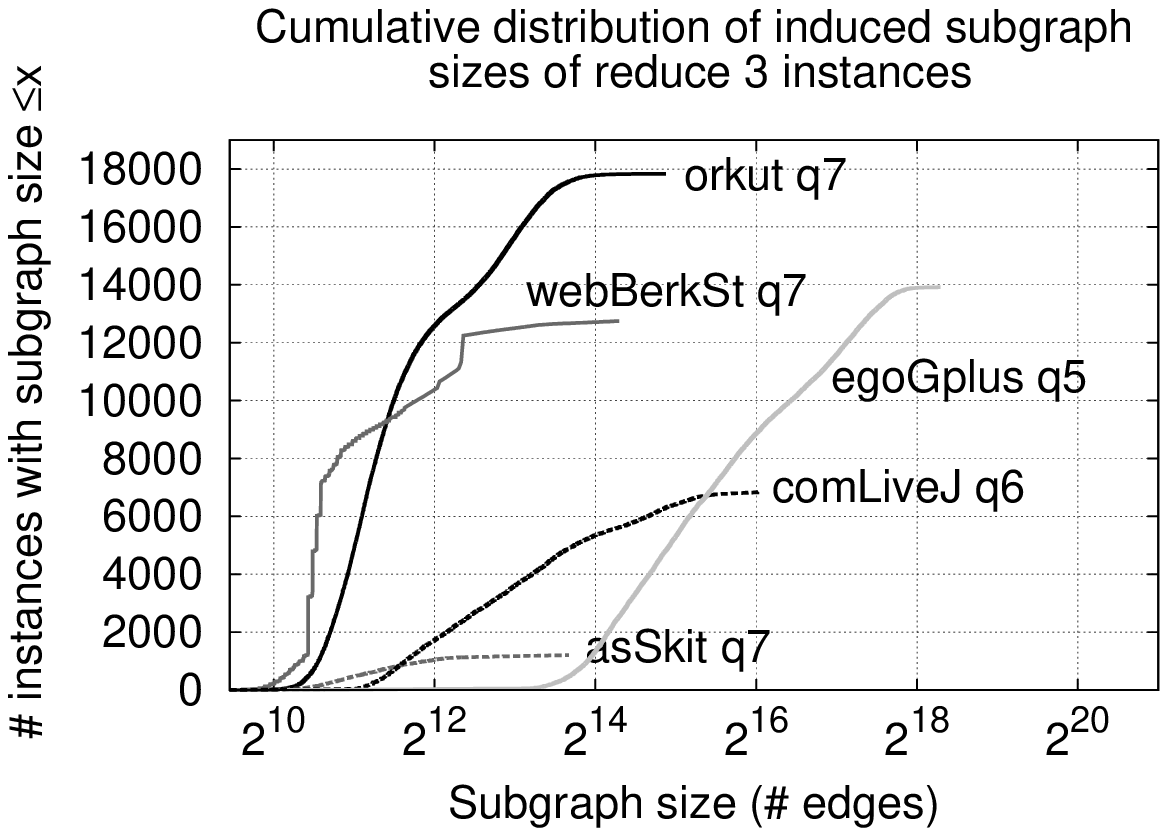} \hspace{3mm} & 
 \raisebox{20mm}{(b)} & \hspace{-3mm} \includegraphics[scale=0.50]{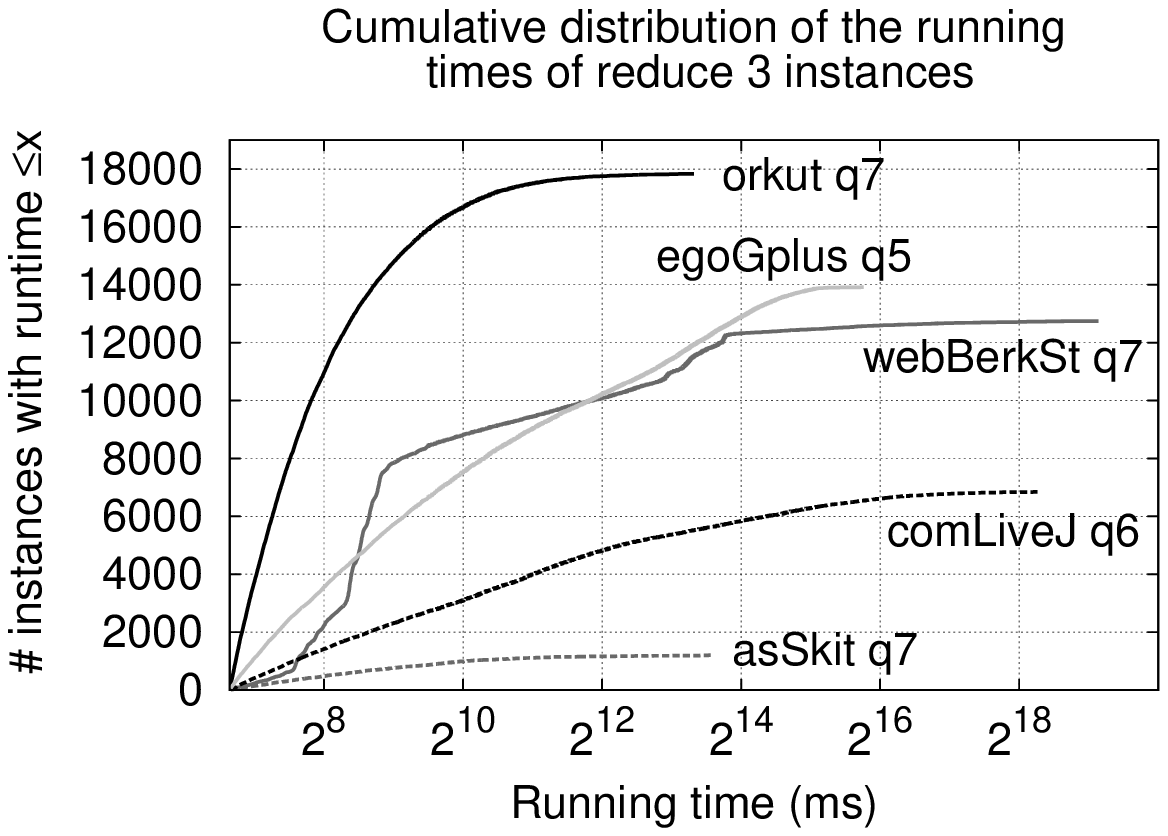} \hspace{-3mm}
\end{tabular}
\vspace{-3mm}
\caption{Analysis of reduce instances of round 3 of algorithm \sik\ on a selection of benchmarks.}
\vspace{-2mm}
\label{fig:round5}
\end{figure*}

\vspace{2mm}
\noindent{\bf Scalability on different clusters.} Since MapReduce algorithms are inherently parallel, a natural question is how their running times are affected by the cluster size (and ultimately, by the available number of cores). Figure~\ref{fig:scalability} exemplifies the running times on three clusters of 4, 8, and 16 nodes. We focus on \sik, which proved to be the algorithm of choice for $k\ge 5$. As an example, the average speedups of \siVI\ when doubling the cluster size from 4 to 8 nodes and from 8 to 16 nodes are 1.44 and 1.53, respectively (the average is taken over the three graph instances). We remark that the maximum theoretical speedup is 2. Similar values can be obtained for the other values of $k$.

While we could not experiment on considerably larger clusters (Amazon AWS limits the number of on-demand instances that can be requested), Figure~\ref{fig:round5} provides some insights. We focus on round 3, which proved to be the most expensive step when $q_k$ is large (see, e.g., \texttt{webBerkStan} and \texttt{comLiveJ} in Figure~\ref{fig:roundAnalysis}). Figure~\ref{fig:round5}b shows the cumulative distribution of the running times of reduce 3 instances (for executions longer than 100 ms). We focused on the least favorable scenarios, corresponding to large values of $k$ for which the problem is more computationally intensive. The rightmost point on each curve gives the runtime of the slowest reduce instance, that reaches 9 minutes when computing $q_7$ on \texttt{webBerkStan}. Although most curves are steeper for short durations, in all cases there are hundreds or even thousands of reduce instances with running times comparable to the slowest one: e.g., in {\tt egoGplus} more than 2000 reducers are within a factor $8\times$ of the slowest one, and even in the $q_7$ computation for \texttt{webBerkStan} 169 instances require more than one minute. This suggests that \sik\ is amenable to further parallelization: we expect that, on a larger cluster, the abundant time-demanding instances could be effectively scheduled to different nodes, yielding globally shorter running times.  This analysis is in line with the distribution of induced subgraph sizes observed in Figure~\ref{fig:round5}a, supporting the conclusion that the harmful ``curse of the last reducer'' phenomenon~\cite{SV11} -- where typically 99\% of the map/reduce instances terminate quickly, but a very long time could be needed waiting for the last task to succeed -- can be kept under control even when $k$ is increased.

\vspace{-3mm}
%-------------------------------------------------------------------------------------------
%-------------------------------------------------------------------------------------------
\section{Approximate counting}
\label{se:approximate}

\vspace{-1mm}

In this section we analyze two variants of a sampling strategy that allows us to decrease the overall space usage, starting from the output of map 2 instances.  The space saving in map 2 instances propagates to the following phases, reducing the space used by reduce 2 as well as map and reduce 3 instances, and also results in an improved running time (due to reduced local complexities and global work).
Instead of performing the sampling directly on the list of edges of the graph, we work by sampling pairs of high-neighbors that are emitted by map 2 instances. 
If each pair of high-neighbors of a given node $u$ is emitted with probability $p$, then each edge of $G^+(u)$ will be included with probability $p$ in the subgraph built by the reduce 3 instance with key $u$; however,  the same edge $e$  in two distinct subgraphs $G^+(u)$ and $G^+(u')$ is sampled independently, which results in improved concentration around the mean.

%------------------------------------------
\begin{figure*}[t]
\centering
\begin{tabular}{ccc}
 \hspace{-9mm}
 \includegraphics[scale=0.45]{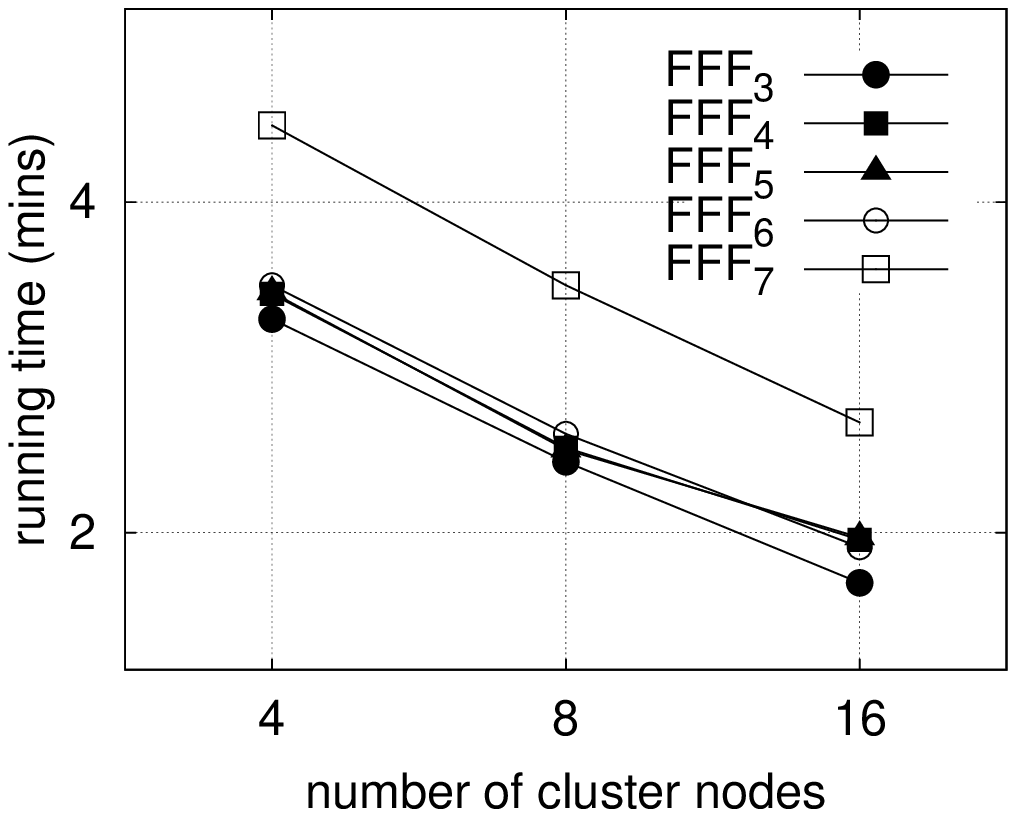} \hspace{-7mm} & 
 \includegraphics[scale=0.45]{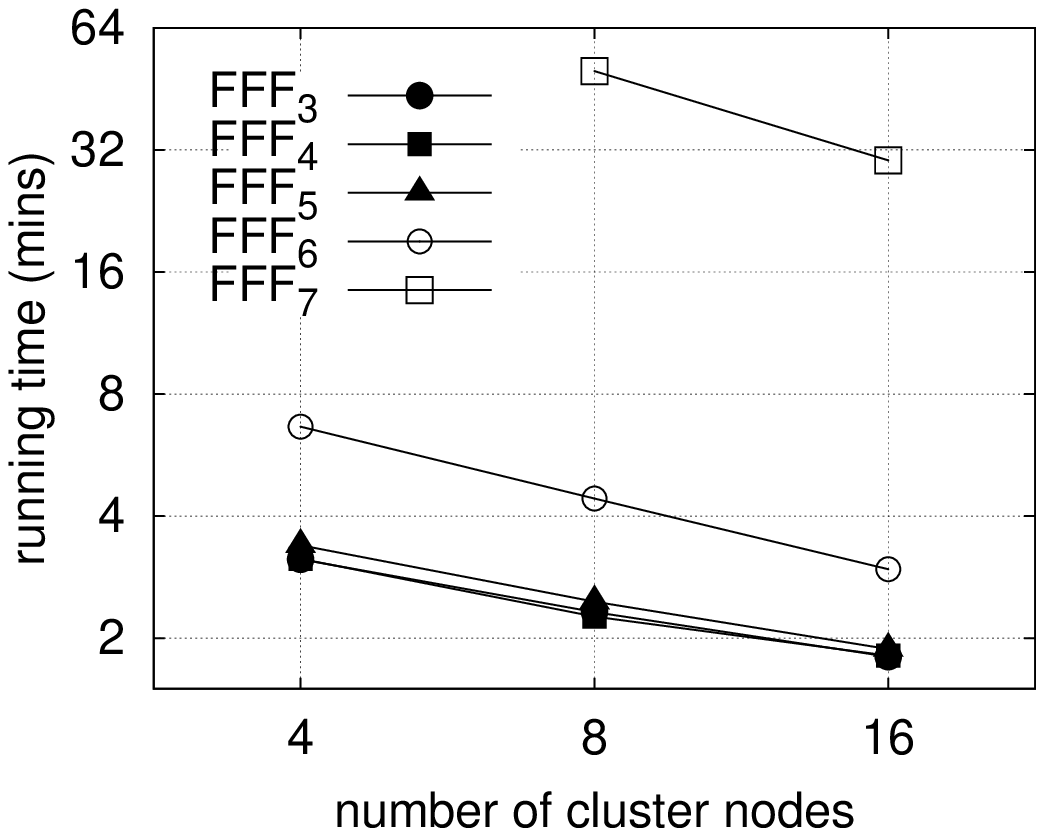} \hspace{-8mm} & 
 \includegraphics[scale=0.45]{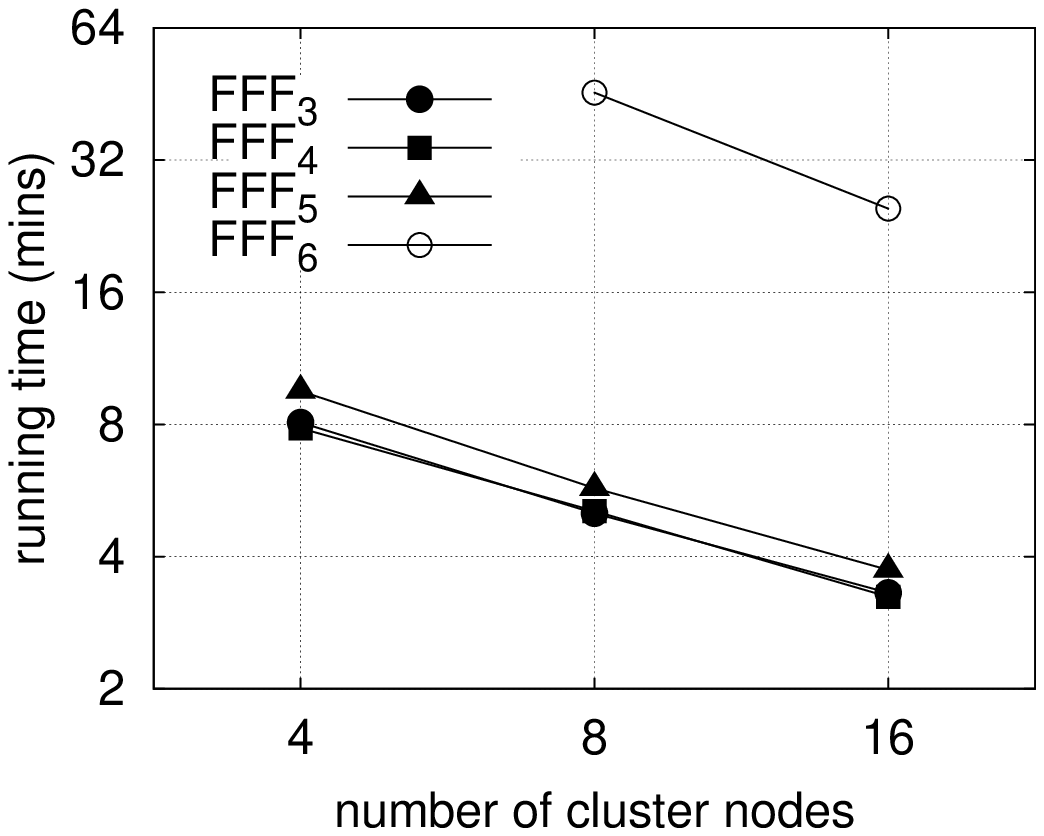}%\\
% (a) \texttt{asSkitter} & (b) \texttt{webBerkStan} & (c) \texttt{liveJournal}
\end{tabular}
\vspace{-4mm}
\caption{Scalability on different cluster sizes: \texttt{asSkit}, \texttt{webBerkStan}, and \texttt{comLiveJ} (left-to-right).}
\vspace{-2mm}
\label{fig:scalability}
\end{figure*}

%\ subsection{Plain pair sampling}
%\ label{ss:plain-sampling}

\vspace{-3mm}
\paragraph{Plain pair sampling.}
We first address the case when pairs of high neighbors are sampled uniformly at random. \full{In details, we modify algorithm \sik\ as follows:
\begin{itemize}[leftmargin=0.4cm,rightmargin=0cm]
\itemsep0em
\item map 2 instances emit the key-value pair $\langle (x_i,x_j);u)$, for all pairs $(x_i,x_j)$ with $x_i \prec x_j$ in $\Gamma^+(u)$,  \emph{with probability $p$}; 
\item  reduce 3 instances emit the pair $\langle u;q_u/p^{(k-1)(k-2)/2})$. 
\end{itemize}
}
In details, map 2 instances emit the key-value pair $\langle (x_i,x_j);u\rangle$, for all pairs $(x_i,x_j)$ with $x_i \prec x_j$ in $\Gamma^+(u)$,  \emph{with probability $p$}, and reduce 3 instances emit the pair $\langle u;q_u/p^{(k-1)(k-2)/2}\rangle$. 
\full{Let $\mathcal{Q} = \{Q_1, \ldots, Q_{q_k}\}$ be the set of $k$-cliques of $G$. Let $Q_i\in \mathcal{Q}$ be a clique and let $u$ be its smallest node (according to $\prec$). We say that $Q_i$ is  \emph{sampled} if all pairs in $Q_i\setminus \{u\}$ were emitted by the map 2 instance with key $u$. }
The following results concerning space usage, correctness, and concentration of the estimate around its mean are proved in Appendix B.

\vspace{-1mm}

\begin{lemma}
\label{le:sampling-space}
Let $p\le 1$ be the edge sampling probability of algorithm \sik\ and let $\alpha$ be a constant in $[0,1)$ such that $1/m^\alpha \le p$. Then algorithm \sik\  with sampling probability $p$ uses local space $O(mp)$ with high probability.
\end{lemma}

\full{We now focus on correctness, proving that the mean value of the estimate of algorithm \sik\ with edge sampling is the number of $k$-cliques of its input graph $G$.}

\full{
\vspace{-3mm}
\begin{lemma}
\label{le:mean}
Algorithm \sik\ with edge sampling returns an estimate $\widetilde{q_k}$ with expected value $E[\widetilde{q_k}]=q_k$, where $q_k$ is the number of $k$-cliques in the input graph $G$.
\end{lemma}
}
\full{We conclude the analysis by providing the conditions under which the estimate $\widetilde{q_k}$ is concentrated around its mean value.}

\vspace{-3mm}
\begin{theorem}\label{th:concentrationEdgeSampling}
Let $G$ be a graph with $m$ edges and $q_k$ $k$-cliques. Let $\widetilde{q_k}$ be the estimate returned by algorithm \sik\ with edge sampling probability $p$.
For any constant $\varepsilon >0$, there exists a constant $h>0$ such that  $|\widetilde{q_k}-q_k| \le \varepsilon q_k$ with high probability if $p^{(k-1)(k-2)/2} >\frac{h m^{(k-3)/2} \ln m}{\varepsilon^2 q_k}$.
\end{theorem}

%-------------------------------------------------------------------------------------------
%\ subsection{Color-based sampling}
%\ label{ss:color-sampling}

\vspace{-6mm}
\paragraph{Color-based sampling.} The authors of~\cite{PT12} proposed a sampling technique that allows to increase the expected number of sampled cliques without increasing the number of sampled edges. This is achieved by coloring the nodes of the graph and sampling all monochromatic edges.
The same idea can be applied in our setting to sample the emitted pairs in map 2 instances by coloring all nodes in each $\Gamma^+(u)$ with $c$ colors and emitting all monochromatic pairs. This has the following implications. Each edge in $G^+(u)$ is sampled with probability $1/c$. A $k$-clique $Q_i$ with smallest node $u$ is sampled with probability $1/c^{k-2}$, given by the probability of assigning all nodes in $Q_i\setminus \{u\}$ with the same color. Hence, reduce 3 instances have to be modified in order to return $\langle u;q_u c^{k-2}\rangle$ as partial estimates of the number of $k$-cliques of $G$. Let us call the resulting approximation algorithm \sick.

The analysis of plain pair sampling can be naturally extended to \sick\ as described in Appendix B. Concentration around the mean is achieved with high probability under the following conditions:\full{expressed by the following theorem:}

\vspace{-1mm}
\begin{theorem}\label{th:concentrationColorSampling}
Let $G$ be a graph of $m$ edges and $q_k$ $k$-cliques.  Let $\widetilde{q_k}$ be the estimate returned by algorithm \sick\ with $c$ colors.
For any constant $\varepsilon >0$ there exist a constant $h>0$, such that $|\widetilde{q_k}-q_k| \le \varepsilon q_k$ with high probability if $1/c^{k-2} >\frac{h m^{k-2} \ln m}{\varepsilon^2 q_k}$.
\end{theorem}

\vspace{-1mm}
For the case $k=3$, Theorem~\ref{th:concentrationColorSampling} guarantees (with high probability) concentration around the mean when \full{$1/c \ge \frac{h m \ln m}{\varepsilon^2 q_3}$} $1/c \ge (h m \ln m)/(\varepsilon^2 q_3)$, which improves the bound in~\cite{PT12}, whose worst-case analysis imposes \full{$1/c^2 \ge \frac{h' n^2 \ln n}{\varepsilon^2 q_3}$} $1/c^2 \ge (h' n^2 \ln n)/(\varepsilon^2 q_3)$ to guarantee concentration  on an $n$-node graph. \full{This is due to the fact that we color the same node $x\in \Gamma^+(u)\cap \Gamma^+(v)$ independently for $u$ and $v$, which allows us to reduce the maximum degree of the interference graph $H$ in the application of the Hajnal-Szemeredi theorem.}

%-------------------------------------------------------------------------------------------
{\vspace{2mm}
\noindent{\bf Discussion.}}
Our estimators fit in the class $\mathcal{MRC}$ as long as $p\le 1/m^\alpha$ (and equivalently $c\ge m^\alpha$), for a small constant $\alpha$, as shown by Lemma~\ref{le:sampling-space}. Moreover, the space reduction translates in improved bounds for the local complexities (in particular for reduce 3 instances that are the most computationally intensive) and work.
Notice that the concentration result in Theorem~\ref{th:concentrationColorSampling} is weaker than that in Theorem~\ref{th:concentrationEdgeSampling}. \full{This is due to the fact that the interference graph used in the proof of Theorem~\ref{th:concentrationEdgeSampling} has an even smaller maximum degree than that resulting from the color-based sampling. As observed above, this is beneficial in the application of the Hajnal-Szemeredi theorem.} However, the color-based sampling strategy increases the expected number of sampled cliques, using sampling probability $p=1/c$, with respect to plain pair sampling. The expected number of sampled cliques shrinks by a factor $p^{(k-1)(k-2)/2}$ for plain sampling, and only by a factor $p^{k-2}$ for color-based sampling. In practice, this boosts the accuracy of the color-based algorithm, which becomes better and better than plain sampling when $k$ grows. We clearly observed this phenomenon, which was also discussed in~\cite{PT12} for triangles, in our experiments.

%------------------------------------------
\begin{table}[t]
\centering
\begin{footnotesize}
\begin{tabular}{|c|ccc|ccc|ccc|crc|ccc|}\hline
 & \multicolumn{3}{c|}{\sic{3}} & \multicolumn{3}{c|}{\sic{4}}  & \multicolumn{3}{c|}{\sic{5}}  & \multicolumn{3}{c|}{\sic{6}}  & \multicolumn{3}{c|}{\sic{7}}\\ 
 \hline
 & {\scriptsize{time}} & {\scriptsize{sp.}} & {\scriptsize{err.}} & {\scriptsize{time}} & {\scriptsize{sp.}} & {\scriptsize{err.}}  & {\scriptsize{time}} & {\scriptsize{sp.}} & {\scriptsize{err.}} & {\scriptsize{time}} & \multicolumn{1}{c}{\scriptsize{sp.}} & {\scriptsize{err.}} & {\scriptsize{time}} & {\scriptsize{sp.}} & {\scriptsize{err.}} \\\hline \hline
{\texttt{asSkit}} & 2:53 & 1.12 & 0.01 & 2:51 & 1.15 & 0.23 & 2:49 & 1.17 & 0.86 & 2:49 & 1.15 & 4.39 & 2.44 & 1.54 & 1.06\\
{\texttt{orkut}} & 6:08 & 3.91 & 0.02 & 5:52 & 3.94 & 0.05 & 6:09 & 3.77 & 0.08 & 5:57 & 3.93 & 0.34 & 6:30 & 4.33 & 2.39\\
{\texttt{webBerkSt}} & 2:45 & 1.09 & 0.05 & 2:46 & 1.09 & 0.16 & 2:47 & 1.13 & 0.17 & 2:42 & 1.83 & 1.22 & 2:44 & 18.4 & 0.39\\
{\texttt{comLiveJ}} & 3:24 & 1.62 & 0.05 & 3:27 & 1.57 & 0.03 & 3:29 & 1.78 & 0.21 & 3:25 & 12.11 & 0.24 & 3:22 & - & -\\
{\texttt{socLiveJ1}} & 3:42 & 1.77 & 0.02 & 3:40 & 1.83 & 0.03 & 3:31 & 2.23 & 0.03 & 3:36 & 24.05 & 0.61 & 3:41 & - & - \\
{\texttt{egoGplus}} & 4:18 & 4.03 & 0.01 & 4:18 & 4.16 & 0.02 & 4:17 & 9.11 & 0.08 & 4:36 & \multicolumn{1}{c}{-} & - & 5:34 & - & -\\
\hline
\end{tabular}
\end{footnotesize}
\vspace{-3mm}
\caption{Running time, speedup, and approximation quality of \sick\ for $k\in [3,7]$. The speedup is with respect to \sik\ (see also Table~\ref{fig:runtime-15nodes}) and the error is given as a percentage.}
\vspace{-3mm}
\label{fig:approximate}
\end{table} 

%\ subsection{Experiments with approximate counting}
%\ label{ss:experiments-approximate}

\vspace{-3mm}
\paragraph{Experiments with approximate counting.}
We experimented with both edge-based and color-based sampling, choosing different sampling probabilities and running each algorithm three times on the same instance and platform configuration to increase the statistical confidence of our results.
We briefly report on the results obtained by algorithm \sick, whose accuracy  in practice outperformed edge sampling. 
As predicted by the theoretical analysis, sampling is beneficial for round 2, since reduces the number of emitted $2$-paths. In turn, this decreases the number of edges in the induced subgraphs constructed at round 3, yielding substantial benefits on the running time of this round: in particular, in all our tests we observed that  the running time of reduce 3 instances of \sick\ remains almost constant as $k$ increases.
Table~\ref{fig:approximate} summarizes the behavior of \sick, showing elapsed time, speedup over the exact algorithm, and accuracy. The experiments were performed on the 16-node cluster using 10 colors, which corresponds to a sampling probability $0.1$. The achieved speedups are dramatic (up to 24$\times$ in \texttt{socLiveJ1}) in all those cases where the exact algorithm took a long time. We were  able to compute in a few minutes the estimated number of $q_6$ and $q_7$ of graphs where the exact computation would have required several hours.
The accuracy is very good, especially on the  datasets that were most difficult for the exact algorithm: the 24$\times$ faster computation on {\tt socLiveJ1}, for instance, returned an estimate that was only 0.61\% away from the exact value of  $q_6$. 

\vspace{-3mm}

%-------------------------------------------------------------------------------------------
\section{Concluding remarks}

\vspace{-1mm}

We have proposed and analyzed, both theoretically and experimentally, a suite of MapReduce algorithms for counting $k$-cliques in large-scale undirected graphs, for any constant $k\ge 3$.  Our experiments, conducted on the Amazon EC2 platform, clearly highlight the algorithm of choice in different scenarios, showing that our algorithms gracefully scale to non-trivial values of $k$, larger instances, and diverse cluster sizes. 
\full{While our exact algorithms can enumerate the number of cliques incident on each node, our approximate solutions return with high probability very accurate estimates of the total number of $k$-cliques, exploiting two different random sampling strategies. All the algorithms work in three rounds. We implemented several variants of our $k$-clique estimators in Hadoop 2.2.0 and tested them on widely-used networks from the SNAP graph library~\cite{SNAP}. The experiments highlight the main performance bottleneck in the exact algorithms, and show how randomization can alleviate this problem, dramatically reducing the running times with a negligible loss of precision and smoothly scaling on different values of $k$, larger instances, and diverse cluster sizes.
Our experimental analysis shows that reduce 3 instances might need to perform rather costly sequential computations when nodes have large high-neighborhoods: when counting $(k-1)$-cliques in subgraphs $G^+(u)$, the clique size $k$ has an immediate impact on the running time of this step and, for $k=7$, we witnessed on one of the instances sequential bursts requiring up to 9 minutes. If we consider the number of nodes $u$ that yield graphs $G^+(u)$ almost as big as the largest node-induced subgraph, this number is fairly large when compared to the size of our clusters: thus, on small cluster sizes our overall running times would not benefit from further parallelization of the $(k-1)$-cliques counting in round 3. Nevertheless ...}
It is worth noticing that our approach could be slightly modified in order to trade overall space usage for local running time. The actual count of $(k-1)$-cliques at round 3 could be indeed postponed for all nodes $u$ such that $G^+(u)$ is too large. In an additional round, map instances would replicate each ``uncounted'' subgraph $G^+(u)$ once per high-neighbor $v$ of $u$, distributing the workload to many reducers. The reduce instance with key $(u,v)$ would thus count the number of $(k-2)$-cliques in its copy of $G^+(u)$. This process could be repeated up to $k-4$ times, before  copying $\sqrt{m}$ times $G^+(u)$ becomes more expensive than counting: each iteration would increase by a factor $\sqrt{m}$ the global space usage and reduce by the same factor the local running times of the reducers, without affecting the total work.  We expect this tradeoff to be rather effective  on very large clusters, especially for skewed distributions of the high-neighborhoods sizes and large values of $k$, and regard assessing its practicality as an interesting direction. 

\newpage

\section*{Acknowledgements}
This work was generously supported by Amazon Web Services through an  AWS in Education Grant Award received by the first author. We are also indebted to Emilio Coppa for many useful discussions about tuning the Hadoop configuration parameters.

\begin{small}
\bibliographystyle{abbrv}

\bibliography{biblio}
\end{small}

%========================================================================
\newpage
\section*{Appendix A: extended related work}

%\vspace{-2mm}
%\paragraph{Related work.}

Many related works focus on triangle counting, which is a fundamental algorithmic problem with a variety of applications. For instance, it is closely related to the computation of the clustering coefficient of a graph, which is in turn a widely used measure in the analysis of social networks~\cite{KMPT12}.
Many exact and approximate algorithms tailored to triangles have been developed in the literature in different computational models. The best known sequential counting algorithm is based on fast matrix multiplication~\cite{AYZ07} and has running time $O(m^{\frac{2\omega}{\omega+1}})$, where $\omega$ is the matrix multiplication exponent: this makes it infeasible even for medium-size graphs. Practical approaches match the $O(m^{3/2})$ bound first achieved in~\cite{IR78} and~\cite{CN85}, which is optimal for the listing problem. As shown in~\cite{OB14}, many listing algorithms hinge upon a common abstraction that also yields the state-of-the-art sequential implementations for the enumeration of triangles. The input/output complexity of the triangle listing problem is addressed in~\cite{PS14}. Approximate counting algorithms that operate in the data stream model, where the input graph is streamed as a list of edges and the algorithm must compute a solution using small space, are presented in~\cite{BYKS02,JG05,BFLMS06,PTTW13,BFLS07}, and a randomization technique to speed-up any triangle counting algorithm while keeping a good accuracy is proposed in~\cite{TKMF09}. 

A few works have addressed counting problems for subgraphs different from (and more difficult than) triangles. For instance, the techniques in~\cite{PTTW13}  also allow to approximate the number of small cliques, the algorithm from~\cite{BFLMS06} can be extended to any subgraph with 3 or 4 nodes~\cite{BDGL08}, while~\cite{MMPS11} tackles the problem of counting cycles. All these works focus on graph streams and return estimates, typically concentrated around the true number with high probability. The more general problem of enumerating arbitrary small subgraphs has been also studied in the data stream model~\cite{KMSS12}. 

%MapReduce exact/approximate algorithms for counting triangles or arbitrary subgraphs are finally presented in~\cite{AFU13,PT12,SV11}, as detailed before.

%========================================================================
\newpage
\section*{Appendix B: proofs}

\subsection*{Tools from the literature}

The following property, which is folklore in the triangle counting literature~\cite{SW05}, will be crucial in the \full{design and} analysis of our clique estimators (we report its short proof for completeness):

\setcounter{theorem}{0}
\renewcommand\thetheorem{\Roman{theorem}}

\begin{theorem}
\label{le:high-neighborhood}
In a graph with $m$ edges,  $|\Gamma^+(u)|\le 2\sqrt{m}$ for each node $u$.
\end{theorem}
\begin{proof}
Let $h$ be the number of nodes with degree larger than $\sqrt{m}$: since there are $m$ edges, it must be $h\le 2\sqrt{m}$.
If $d(u)>\sqrt{m}$, then nodes in $\Gamma^+(u)$ must also have degree larger than $\sqrt{m}$ and their number is upper bounded by $h\le 2\sqrt{m}$.
If $d(u)\le\sqrt{m}$, the claim trivially holds since $\Gamma^+(u)\subseteq\Gamma(u)$. 
\end{proof}

In the analysis of our approximate estimators we make use of the following (weaker) version of the Chernoff concentration inequality~\cite{C81}:

\begin{theorem}
\label{le:chernoff}
Let $X_1, \ldots, X_h$ be independent identically distributed Bernoulli random variables, with probability of success  $p$. Let $X=\sum_{i=1}^{h} X_i$ be a random variable with expectation $\mu=p\cdot h$. Then, for any $\varepsilon\in(0,1)$, $Pr\{|X-\mu|>\varepsilon\mu\}\le 2e^{-\varepsilon^2\mu/3}$.
\end{theorem}

We will also exploit the following conjecture of Paul Erd\"os, proved in 1970 by Hajnal and Szemeredi~\cite{HS70}:

\begin{theorem}
\label{le:hajnal-szemeredi}
Every $n$-node graph with maximum degree $\Delta$ is $(\Delta + 1)$-colorable with all color classes of size at least $n/\Delta$.
\end{theorem}

We say that an event has high probability when it happens, for a graph $G$ of $m$ nodes, with probability at least $1-1/m$, for large enough $m$.

\subsection*{Exact counting }

\renewcommand\thetheorem{\arabic{theorem}}
\setcounter{lemma}{0}
\setcounter{theorem}{0}

\begin{theorem}
Let $G$ be a graph and let $m$ be the number of its edges.
Algorithm \sik\ counts the $k$-cliques of $G$ using $O(m^{3/2})$ total space and $O(m^{k/2})$  total work. The local space and the local running time of mappers and reducers are $O(m)$ and $O(m^{(k-1)/2})$, respectively.
\end{theorem}
\begin{proof}
The total space usage in round 1 is $O(m)$. Map 2 instances produce key-value pairs of constant size, whose total number is upper bounded by $\sum_{u\in V}{{|\Gamma^+(u)|}\choose{2}}$, which is at most $2\sqrt{m}\cdot \sum_{u\in V}|\Gamma^+(u)|=O(m^{3/2})$ by Theorem~\ref{le:high-neighborhood}. The data volume can only decrease after the execution of reduce 2 instances and is not affected by round 3. Hence, the total space usage is $O(m^{3/2})$.

We now consider local space. Map 1 instances use constant memory. By Theorem~\ref{le:high-neighborhood}, the input to any reduce 1 instance has size $O(\sqrt m)$. Similarly, any map 2 instance receives $O(\sqrt m)$ input edges and produces $O(m)$ key-value pairs. Consider a reduce 2 instance and let $(x,y)$ be its key. The input of this instance is $\Gamma^-(x)\cap\Gamma^-(y)\subseteq V$ (without any repetition), and its size is thus $O(n)$. In round 3, map and reduce instances use memory $O(n)$ and $O(m)$, respectively, which concludes the proof of the local space claim (recall that $n\le 2m$).

By similar arguments, the running time of map instances is $O(1)$, $O(m)$, and $O(n)$, respectively, in the three rounds. Reduce instances require time $O(\sqrt{m})$ and $O(n)$ in rounds 1 and 2, while reduce 3 instances run on graphs of at most $\sqrt{m}$ nodes and require $O(m^{(k-1)/2})$ time.
The total work of the algorithm is dominated by the costs of the reducers of the last phase, which is upper bounded by $O(\sum_{u\in V}|\Gamma^+(u)|^{k-1})$. By Theorem~\ref{le:high-neighborhood}, this is $O(m^{(k-2)/2}\sum_{u\in V}|\Gamma^+(u)| = O(m^{k/2})$.

Each clique is counted exactly once by the reducer associated to its minimum node (according to $\prec$), proving the correctness of the algorithm. 
\end{proof}

\subsection*{Approximate counting }

%\setcounter{lemma}{3}
%\smallskip

We now prove the results about space usage, correctness, and concentration of the plain pair sampling algorithm described in Section~\ref{se:approximate}. Let $\mathcal{Q} = \{Q_1, \ldots, Q_{q_k}\}$ be the set of $k$-cliques of $G$. Let $Q_i\in \mathcal{Q}$ be a clique and let $u$ be its smallest node (according to $\prec$). We say that $Q_i$ is  \emph{sampled} if all pairs in $Q_i\setminus \{u\}$ were emitted by the map 2 instance with key $u$. 

\begin{lemma}
Let $p\le 1$ be the edge sampling probability of algorithm \sik\ and let $\alpha$ be a constant in $[0,1)$ such that $1/m^\alpha \le p$. Then algorithm \sik\  with sampling probability $p$ uses local space $O(mp)$ with high probability.
\end{lemma}
\begin{proof}
We will prove the claim for reduce 3 instances; similar arguments can be used to prove the bounds for reduce 2 and map 3 instances, recalling that $n\le 2m$.

By Theorem~\ref{le:high-neighborhood} each subgraph $G^+(u)$ has at most $2\sqrt{m}$ nodes.
Since the reduce 3 instance with key $u$ receives an edge $(x,y)\in G^+(u)$ if and only if the map 2 instance with key $u$ sampled the pair $(x,y)$ (event that has probability $p$), we have that the expected input size of any reduce 2 instance is at most $2pm$.

Being each pair of high-neighbors of a node $u$ sampled independently by all the other, a simple application of the Chernoff bound allows to prove that  the probability of one reduce 3 instance to receive more than $4pm$ values is less than $e^{-2pm/3}$, and a union bound gives that the probability of any of the reduce 3 instances to receive more than $4pm$ values is bounded by $n/e^{2pm/3}$. Since $p\ge 1/m^\alpha$, for large enough $m$ this probability is smaller than $1/m$, which concludes the proof.
\end{proof}

\smallskip

\begin{claim}
\label{le:mean}
Algorithm \sik\ with edge sampling returns an estimate $\widetilde{q_k}$ with expected value $E[\widetilde{q_k}]=q_k$, where $q_k$ is the number of $k$-cliques in the input graph $G$.
\end{claim}
\begin{proof}
Let $Q_i\in \mathcal{Q}$ be a clique in $G$ and let $u$ be its smallest node (according to $\prec$).
Clique $Q_i$ contributes to the estimate of the number of cliques in $G$ if the map 2 instance handling input $\langle u; \Gamma^+(u)\rangle$ emits all pairs of nodes in $Q_i \setminus {u}$, i.e., if it is sampled. Let $X_i$ be the random variable indicating the event ``the clique $Q_i$ is sampled''. Since each pair is sampled by the algorithm independently with probability $p$ and there are ${k-1}\choose 2$ distinct (unordered) pairs in a set of $k-1$ elements, $Pr\{X_i=1\}=p^{(k-1)(k-2)/2}$, hence $E[X_i]=p^{(k-1)(k-2)/2}$. 
Now we can define $X=\sum_{i=1}^{q_k} X_i$, and by linearity of expectation we have that $E[X]=q_k p^{(k-1)(k-2)/2}$. Since $\widetilde{q_k}=X / p^{(k-1)(k-2)/2}$ the claim follows.
\end{proof}

\smallskip

\begin{theorem}
Let $G$ be a graph with $m$ edges and $q_k$ $k$-cliques. Let $\widetilde{q_k}$ be the estimate returned by algorithm \sik\ with edge sampling probability $p$.
For any constant $\varepsilon >0$, there exists a constant $h>0$ such that  $|\widetilde{q_k}-q_k| \le \varepsilon q_k$ with high probability if $p^{(k-1)(k-2)/2} >\frac{h m^{(k-3)/2} \ln m}{\varepsilon^2 q_k}$.
\end{theorem}
\begin{proof}
We proved that the expected value of $\widetilde{q_k}$ is $q_k$ in Claim~\ref{le:mean}; we will now deal with the concentration of $\widetilde{q_k}$ around its mean.
Let $H$ be the graph defined as follows:
\begin{itemize}[leftmargin=0.4cm,rightmargin=0cm]
\itemsep0em
\item The node set of $H$ is the set of $k$-cliques $\mathcal{Q}$ of $G$.
\item Let $Q_i$ and $Q_j$ be two $k$-cliques with the same smallest node $u$: there is an edge between $Q_i$ and $Q_j$ in $H$ if and only if $Q_i \setminus \{u\}$ and $Q_j\setminus \{u\}$ have at least one common edge.
\end{itemize}
Cliques that are adjacent in $H$ must thus share at least three of their $k$ nodes, including the smallest node. Considering that graphs $G^+(u)$ have at most $\sqrt{m}$ nodes, we have that the maximum degree of a node in $H$ is therefore $O(m^{(k-3)/2})$. 

Theorem~\ref{le:hajnal-szemeredi} by Hajnal and Szemeredi implies that there exists a node coloring of $H$, using $C\in O(m^{(k-3)/2})$ colors, such that each monochromatic set of nodes has size $\Theta(q_k/m^{(k-3)/2})$, since $H$ has $q_k$ nodes.

For each $j\in [1,q_k]$, let $X_j$ be the indicator variable that is 1 when $Q_j$ is sampled. Let $S_1, \ldots, S_{C}$ be the sets of monochromatic nodes in $H$ and, for each $i\in [1,C]$, let 
\begin{equation}
\label{eq:sum-indicators}
X_{S_i}=\sum_{j: Q_j\in S_i}X_j
\end{equation} 
The terms of  $X_{S_i}$ are independent. Hence, we can apply to $X_{S_i}$ the Chernoff bound as in Theorem~\ref{le:chernoff} obtaining:
$$Pr\{|X_{S_i} - \mu_i| > \varepsilon \mu_i \} \le 2e^{-\mu_i\varepsilon^2/3}\,.$$
where $\mu_i=E[X_{S_i}]$. By Equation~\ref{eq:sum-indicators}, for each set $S_i$ we have:
$$\mu_i \in \Theta\left(\frac{p^{(k-1)(k-2)/2} q_k}{m^{(k-3)/2}}\right)$$ since $E[X_j]=p^{(k-1)(k-2)/2}$, as shown in the proof of Claim~\ref{le:mean}. By the same proof,  $\mu = \sum_{i=1}^C \mu_i = p^{(k-1)(k-2)/2} q_k$. 
 
If we define $X=\sum_{i=1}^C X_{S_i}$ and we apply the union bound we can conclude that
$$Pr\{|X- \mu|>\varepsilon \mu\} \le c_1 m^{\frac{k-3}{2}} e^{-\frac{c_2 \varepsilon^2 q_k p^{(k-1)(k-2)/2}}{m^{(k-3)/2}}}$$ by appropriately choosing constants $c_1$ and $c_2$.
By imposing $$c_1 m^{\frac{k-3}{2}} e^{- \frac{c_2 \varepsilon^2 q_k p^{(k-1)(k-2)/2}}{m^{(k-3)/2}}}\le\frac{1}{m}$$ the claim follows with standard algebraic calculations. 
\end{proof}

The analysis above can be naturally extended to algorithm \sick. Lemma~\ref{le:sampling-space} holds for algorithm \sick\ just by using $p=1/c$. The estimate returned by algorithm \sick\ has expected value $q_k$, and this can be proved using arguments similar to those in the proof of Claim~\ref{le:mean}.
The arguments of the proof of Theorem~\ref{th:concentrationEdgeSampling} can also be used to prove concentration around the mean for algorithm \sick, considering that correlation of sampled $k$-cliques arises as soon as the cliques share a node besides the minimum node, instead of an edge. Hence, concentration is achieved with high probability under the conditions expressed by Theorem~\ref{th:concentrationColorSampling}.

We observed in Section~\ref{se:approximate} that, for $k=3$, our concentration results require weaker conditions than the triangle counting algorithm from~\cite{PT12}. This is due to the fact that we color the same node $x\in \Gamma^+(u)\cap \Gamma^+(v)$ independently for $u$ and $v$, which allows us to reduce the maximum degree of the interference graph $H$ in the application of the Hajnal-Szemeredi theorem.

%========================================================================
\newpage
\section*{Appendix C: detailed benchmark statistics and number of cliques}

\begin{center}

\vspace{0.5cm}
%\begin{table}[h]
{\centering{
\scalebox{0.90}{\rotatebox{90}{
                \begin{minipage}{22.5cm}
\begin{tabular}{|c|crrrrrr|}\hline
& \multicolumn{1}{c}{$n$ ($=q_1$)} & \multicolumn{1}{c}{$m$ ($=q_2$)} & \multicolumn{1}{c}{$q_3$} & \multicolumn{1}{c}{$q_4$} & \multicolumn{1}{c}{$q_5$} & \multicolumn{1}{c}{$q_6$} & \multicolumn{1}{c|}{$q_7$} \\\hline\hline
\texttt{citPat} & 3\,774\,768 & 16\,518\,947 & 7\,515\,023 & 3\,501\,071 & 3\,039\,636 & 3\,151\,595 & 1\,874\,488\\
 (264.0 MB)  &  & ($4.38\times$) & ($0.45\times$) & ($0.47\times$) & ($0.87\times$) & ($1.04\times$) & ($0.6\times$) \\\hline
\texttt{youTube}  & 1\,134\,890 & 2\,987\,624 & 3\,056\,386 & 4\,986\,965 & 7\,211\,947 & 8\,443\,803 & 7\,959\,704 \\
(38.7 MB)  &  & ($2.63\times$) & ($1.02\times$) & ($1.63\times$) & ($1.44\times$) & ($1.17\times$) & ($0.94\times$) \\\hline
\texttt{locGowalla} & 196\,591 & 950\,327 & 2\,273\,138 & 6\,086\,852 & 14\,570\,875 & 28\,928\,240 & 47\,630\,720\\
(11.1 MB)  &  & ($4.83\times$) & ($2.39\times$) & ($2.67\times$) & ($2.39\times$) & ($1.98\times$) & ($1.65\times$) \\\hline
\texttt{socPokec} & 1\,632\,803  &22\,301\,964  & 32\,557\,458 & 42\,947\,031 & 52\,831\,618 & 65\,281\,896 &  83\,896\,509\\
(309.1 MB)  &  & ($13.65\times$) & ($1.46\times$) & ($1.32\times$) & ($1.23\times$) & ($1.18\times$) & ($1.28\times$) \\\hline
\texttt{webGoogle} & 875\,713 & 4\,322\,051 & 13\,391\,903 & 39\,881\,472 & 105\,110\,267 & 252\,967\,829& 605\,470\,026\\
(59.5 MB)  &  & ($4.93\times$) & ($3.1\times$) & ($2.98\times$) & ($2.63\times$) & ($2.40\times$) & ($2.39\times$) \\\hline
\texttt{webStan}  & 281\,903 & 1\,992\,636 & 11\,329\,473 & 78\,757\,781 & 620\,210\,972 & 4\,859\,571\,082 & 34\,690\,796\,481\\
(26.4 MB)  &  & ($7.07\times$) & ($5.68\times$) & ($6.95\times$) & ($7.87\times$) & ($7.83\times$) & ($7.13\times$) \\\hline
\texttt{asSkit}  & 1\,696\,415 & 11\,095\,298 & 28\,769\,868 & 148\,834\,439 & 1\,183\,885\,507 & 9\,759\,000\,981 & 73\,142\,566\,591 \\
 (149.1 MB)  &  & ($6.54\times$) & ($2.59\times$) & ($5.17\times$) & ($7.95\times$) & ($8.24\times$) & ($7.49\times$) \\\hline
\texttt{orkut}  & 3\,072\,441 & 117\,185\,083 & 627\,584\,181 & 3\,221\,946\,137 & 15\,766\,607\,860& 75\,249\,427\,585&353\,962\,921\,685 \\
(1\,687.8 MB)  &  & ($38.14\times$) & ($5.35\times$) & ($5.13\times$) & ($4.89\times$) & ($4.77\times$) & ($4.70\times$) \\\hline
\texttt{webBerkStan}  & 685\,230 & 6\,649\,470 & 64\,690\,980 & 1\,065\,796\,916 & 21\,870\,178\,738 & 460\,155\,286\,971& 9\,398\,610\,960\,254\\
(89.4 MB)  &  & ($9.70\times$) & ($9.73\times$) & ($16.48\times$) & ($20.52\times$) & ($21.04\times$) & ($20.42\times$) \\\hline
\texttt{comLiveJ}  & 3\,997\,962 & 34\,681\,189 & 177\,820\,130 & 5\,216\,918\,441 & 246\,378\,629\,120& 10\,990\,740\,312\,954& {\it 445\,377\,238\,737\,777} \\
(501.6 MB)  &  & ($8.67\times$) & ($5.12\times$) & ($29.34\times$) & ($47.22\times$) & ($44.6\times$) & {\it(40.52$\times$)} \\\hline
\texttt{socLiveJ1}  & 4\,847\,571 & 42\,851\,237 & 285\,730\,264 & 9\,933\,532\,019 &467\,429\,836\,174 & 20\,703\,476\,954\,640 & {\it 849\,206\,163\,678\,934}\\
(627.7 MB)  &  & ($8.84\times$) & ($6.67\times$) & ($34.7\times$) & ($47.05\times$) & ($44.29\times$) & {\it(41.01$\times$)} \\\hline
\texttt{egoGplus} & 107\,614 & 12\,238\,285 & 1\,073\,677\,742 & 78\,398\,980\,887 & 4\,727\,009\,242\,306 & {\it 242\,781\,609\,271\,577}& {\it 11\,381\,161\,386\,691\,540}\\
(538.5 MB)  &  & ($113.72\times$) & ($87.73\times$) & ($73.01\times$) & ($60.29\times$) & {\it (51.36$\times$)} & {\it(46.87$\times$)} 
\\\hline\hline
\end{tabular}
%\caption{
\\~\\~\\
Benchmark statistics: number $n$ of nodes, number $m$ of edges, numbers of cliques on $3$, $4$, $5$, $6$, and $7$ nodes (clique numbers in {\it italic} are approximations obtained by our color-based sampling algorithm, using 10 colors). Benchmarks are sorted by increasing $q_7$. For each benchmark, we also report in parentheses the storage in MB with no compression and the ratio $q_{k+1}/q_k$ (which is half the average node degree for $k=1$). Notice the large values of these ratios for benchmarks \texttt{socLiveJ1}, \texttt{comLiveJ}, \texttt{webBerkStan}, and \texttt{egoGplus}. 
%}
%\label{fig:data-sets-detailed}
\end{minipage}}}
}}
%\end{table}

\end{center}

\end{document}